\documentclass[11pt]{amsart}

\usepackage{hyperref}
\usepackage{amscd}
\usepackage{amssymb}
\usepackage{amsthm}

\theoremstyle{plain}
\newtheorem{Thm}{Theorem}
\newtheorem{Prop}{Proposition}[section]
\newtheorem{Lem}{Lemma}[section]
\newtheorem{Def}{Definition}

\theoremstyle{remark}
\newtheorem{Rem}{Remark}[section]

\numberwithin{equation}{section}

\begin{document}
\author[P. Busch]{P. ~Busch\\
$\text{\rm\tiny Former  Affiliation: Department of Mathematics, University of Hull, UK}$\\
$\text{\rm\tiny Current: Department of Mathematics, University of York, UK}$\\
$\text{\rm\tiny E-mail address: {\tt paul.busch@york.ac.uk}}$\\
\hspace{1cm}\hfill\\
$\text{\rm\tiny Published in: {\em Mathematical Physics, Analysis and Geometry} {\bf 2} (1999) 83-106.}$\\
DOI: \href{http://dx.doi.org/10.1023/A:1009822315406}{10.1023/A:1009822315406}
}
\title{Stochastic Isometries in Quantum Mechanics}
\thanks{
The author wishes to thank Pekka Lahti for helpful comments on an earlier
version of the manuscript.}
\date{}
\maketitle

\begin{abstract}
The class of stochastic maps, that is, linear, trace-preserving, positive
maps between the self-adjoint trace class operators of complex separable
Hilbert spaces plays an important role in the representation of reversible
dynamics and symmetry transformations. Here a characterization of the
isometric stochastic maps is given and possible physical applications are
indicated.
\end{abstract}

\section{Introduction}

The mathematical modeling of a quantum dynamical system is based on the dual
concepts of states and observables. Of particular importance for the
representation of dynamics and symmetries are structure-preserving bijective
maps of the sets of states and observables, respectively. Such maps,
referred to as affine automorphisms of states and Jordan automorphisms of
observables, respectively, have been quite thoroughly studied in the Hilbert
space model (\cite{Bar}, for a recent systematic account, see \cite{Cas})
and in the $C^{*}$-algebraic formulation of quantum theory \cite{Kad}. These
automorphisms are special instances of the larger classes of linear
isometries acting on the state space and the algebra of observables,
respectively, which are of interest in their own right. While isometries of
operator algebras were analyzed by Kadison in 1951 \cite{Kad2}, applications
of isometric transformations of states have been considered only in recent
years. The characterization of the class of isometric state transformations
in the Hilbert space model of quantum mechanics is the subject of the
present paper.

The paper is organized as follows. In Section 2 the notion of a linear state
transformation -- also called \textsl{stochastic map} -- in the Hilbert
space model of quantum mechanics is presented and some of its basic
properties are reviewed. An isometric map among the stochastic maps is
distinguished by a feature that turns out to be fundamental for the
subsequent considerations: it sends orthogonal pairs of states to orthogonal
pairs (Proposition \ref{pro2-orth}). In Section 3 the main result (Theorem 
\ref{th3-is}) concerning isometric stochastic maps is stated: any stochastic
isometry decomposes into a convex combination of pure stochastic isometries
onto mutually orthogonal ranges. In the case of a completely positive
stochastic isometry, all these pure isometries are induced by unitary maps
(Theorem \ref{th3-cp}). The decomposition of stochastic isometries involves
two steps (Propositions \ref{pro3-dec} and \ref{pro3-mix}), proofs of which
are carried out in Sections 4 and 5. In the first step one is led to
introduce the concept of \textsl{mixing} isometries (Definition 
\ref{def3-mix}), while the second step consists of an analysis of the structure
of mixing isometries.

Section 6 concludes with some general observations and an outline of
physical applications. First, a stochastic isometry is shown to allow an
interpretation as a model of a reversible state transformation even when
surjectivity is not given (Theorem \ref{th6-inv}). An isometric state
transformation leads to a reduction of the state space so that some
observables and symmetries can no longer be distinguished. This may be taken
as a reversible model of structure formation. Finally, a stochastic isometry
can be used in conjunction with a certain transformation of observables to
produce equivalent descriptions of the same physical system; these occur
naturally in the construction of extended models accounting for new
experiences with the given system.

\section{Stochastic Maps}

Let $\mathcal{H}$ and $\tilde{\mathcal{H}}$ be complex separable Hilbert
spaces with inner products $\left\langle \,\cdot \,|\,\cdot \,\right\rangle $%
, taken to be linear in the second argument and conjugate linear in the
first. Let $\mathcal{B}_1(\mathcal{H})$ be the complex Banach space of
linear operators of trace class on $\mathcal{H}$, with trace functional $%
\text{tr}[\cdot ]$ and trace norm $||\cdot ||_1$. Its dual space can be
identified with the Banach space $\mathcal{B}(\mathcal{H})$ of bounded
linear operators on $\mathcal{H}$. The self-adjoint\ parts of $\mathcal{B}_1(%
\mathcal{H})$ and $\mathcal{B}(\mathcal{H})$, $V:={\mathcal{B}}_1
({\mathcal{H}}%
)_s$ and $W:={\mathcal{B}}({\mathcal{H}})_s$, are real Banach spaces. By $V^{+}$
we denote the set of positive linear operators in $V$, the convex positive
cone of $V$. With reference to the quantum physical application, $V$ is
called \textsl{state space}, the elements of the subset ${\mathcal{S}}%
:=V^{+}\cap \{\rho \in V:\text{tr}[\rho ]=1\}$ are called \textsl{states}.
The set $\mathcal{S}$ is convex, its extreme points are given by the
orthogonal projections of rank 1. These are called \textsl{pure states} and
denoted $P_{\left. \varphi \right. }$, $\varphi \in \mathcal{H}$, $\varphi
\ne 0$. For $\tilde{\mathcal{H}}$, we denote the corresponding entities as $%
\tilde{V},$ $\tilde{W},\tilde{\mathcal{S}}$, etc. A linear map $T:V\to 
\tilde{V}$ is called a \textsl{stochastic map} (on $V$) if it is \textsl{%
positive} [$T(V^{+})\subset \tilde{V}^{+}$] and \textsl{trace preserving} [${%
\text{tr}}\circ T=\text{tr}$], or equivalently, if it sends states to states
[$T(\mathcal{S})\subseteq \tilde{\mathcal{S}}$]. A stochastic map that is
isometric [$||T(\rho )||_1=||\rho ||_1$, $\rho \in V$] will be called 
\textsl{stochastic isometry}.

The term operator will be used throughout to denote a linear operator. By $0$
and $I$ (or $I_{\mathcal{H}}$) we will denote the zero and identity
operators, respectively. According to the spectral theorem, any
self-adjoint\ operator\ $a$ admits a decomposition into positive and
negative parts, denoted $a=a_{+}-a_{-}$, where the ranges of $a_{+},a_{-}$
are mutually orthogonal subspaces. For a bounded operator $a$ we define $%
|a|:=(a^{*}a)^{1/2}$. If $a$ is self-adjoint, one has $|a|=a_{+}+a_{-}$. For 
$\rho \in V$ we have $\left\| \rho \right\| _1=\text{tr}\bigl[\rho _{+}\bigr]%
+\text{tr}\bigl[\rho _{-}\bigr]=\left\| \rho _{+}+\rho _{-}\right\| _1$, and
so 
\begin{equation}
\rho \in V^{+}\ \Longleftrightarrow \ \left\| \rho \right\| _1=\text{tr}[%
\rho ].  \label{eq2-pos}
\end{equation}

We recall some basic facts about stochastic maps and stochastic isometries.
A linear map $T:V\to \tilde{V}$ is called \textsl{contractive} [or a \textsl{%
contraction}], if for all $\rho \in V$, $\left\| T(\rho )\right\| _1\le
\left\| \rho \right\| _1$.

\begin{Lem}
\label{le2-stoch}Let $T:V\to \tilde{V}$ be a linear map. The following are
equivalent:\newline
{(i)}\ \ $T$ is trace preserving and positive (i.e. a stochastic map);%
\newline
{(ii)}\ $T$ is trace preserving and contractive;\newline
{(iii)} $T(\mathcal{S})\subseteq \tilde{\mathcal{S}}$.
\end{Lem}

\noindent 
The proof is a straightforward application of Eq. (\ref{eq2-pos}).

\begin{Def}
\label{def2-orth}Positive elements $\rho ,\sigma \in V^{+}\backslash \{0\}$
are called \textsl{orthogonal}, $\rho \perp \sigma $, if $\rho \cdot \sigma
=0$. A stochastic map $T:V\rightarrow \tilde{V}$ is said to be \textsl{%
orthogonality preserving} if $\rho \perp \sigma $ implies $T(\rho )\perp
T(\sigma )$ for all $\rho ,\sigma \in \mathcal{S}$ [and thus for all $\rho
,\sigma \in V^{+}\backslash \{0\}$].
\end{Def}

By the spectral theorem, $\rho \perp \sigma $ exactly when $\rho $ and $%
\sigma $ are the positive and negative parts of $\rho -\sigma $,
respectively. Hence we have: 
\begin{equation}
\rho \perp \sigma \ \Longleftrightarrow \ \left\| \rho -\sigma \right\|
_1=\left\| \rho +\sigma \right\| _1,\ \ \rho ,\sigma \in V^{+}\backslash
\{0\}.  \label{eq2-orth}
\end{equation}

\begin{Prop}
\label{pro2-orth}Let $T:V\rightarrow \tilde{V}$ be a stochastic map. The
following are equivalent: \newline
{(i)}\ \ $T$ is an isometry;\newline
{(ii)}\ $T$ is orthogonality preserving;\newline
{(iii)} $T$ is orthogonality preserving for pairs of pure states.
\end{Prop}

\noindent 
The proof is a straightforward application of Eq. (\ref{eq2-orth}).

\section{Stochastic Isometries\label{sec-stochis}}

In this section the main results concerning stochastic isometries will be
stated. The proofs will be developed in Sections 4 and 5.

A stochastic map is called \textsl{pure} if it sends pure states to pure
states. An isometric linear or antilinear map $U:\mathcal{H}\rightarrow 
\tilde{\mathcal{H}}$ will be called \textsl{unitary} or \textsl{antiunitary}%
, respectively. We first present the known case of surjective stochastic
isometries.

\begin{Prop}
\label{pro3-pure}Let $T:V\rightarrow \tilde{V}$ be a surjective stochastic
isometry. Then $T$ is of the form 
\begin{equation}
T(\rho )\ =\ U\rho \,U^{*},\ \ \rho \in V,  \label{eq3-pure}
\end{equation}
where $U:\mathcal{H}\rightarrow \tilde{\mathcal{H}}$ is unitary or
antiunitary.
\end{Prop}

\begin{proof}
Noting that the inverse of a surjective stochastic isometry $T$ is a trace
preserving isometric map on $V$ and hence, by Lemma \ref{le2-stoch}, a
stochastic isometry, it is easy to verify that $T$ is pure. By Theorem 2.3.1
of \cite{Dav} it follows that $T$ is of the form $T(\rho )=U\rho \,U^{*}$, $%
\rho \in V$, where $U$ is linear or antilinear. Furthermore, $U^{*}U=I_{%
\mathcal{H}}$ since $T$ is trace-preserving, and $UU^{*}=I_{\tilde{\mathcal{H%
}}}$ by surjectivity of $T$. Hence $U$ is unitary or antiunitary.%
\end{proof}

If $\mathcal{H}=\tilde{\mathcal{H}}$ is finite dimensional then any
stochastic isometry $T:V\rightarrow V$ is surjective, hence pure and of the
form (\ref{eq3-pure}).

An affine automorphism of the convex set of states $\mathcal{S}$ is a
bijective affine map of $\mathcal{S}$ onto itself. Such maps are taken to
represent symmetries on the set of states. Any affine map on $\mathcal{S}$
has a unique extension to a linear map on $V$, which is a stochastic map. An
affine automorphism of $\mathcal{S}$ extends to a surjective stochastic
isometry. Thus it is seen that in the case $\mathcal{H}=\tilde{\mathcal{H}},$
Proposition \ref{pro3-pure} reproduces the Wigner-Kadison characterization
of symmetries \cite{Bar},\cite{Kad},\cite{Cas}.

\begin{Prop}
\label{pro3-is}Let $\mathcal{H},\tilde{\mathcal{H}}$ be complex separable
Hilbert spaces such that $\tilde{\mathcal{H}}$ can be presented as a direct
sum of mutually orthogonal closed subspaces, $\tilde{\mathcal{H}}=\left(
\oplus _{k=1}^N\tilde{\mathcal{H}}_k\right) \oplus \tilde{\mathcal{H}}_0$, $%
N\in \Bbb{N}\cup \{\infty \}$, with $\mathrm{\dim }\tilde{\mathcal{H}}_k=%
\mathrm{\dim }\mathcal{H}$, $k=1,2,\cdots ,N$. Let $U_k:{\mathcal{H}}\to 
\tilde{\mathcal{H}}_k$ be unitary or antiunitary maps, $w_1,w_2,\cdots \in
(0,1)$, $\sum_kw_k=1$. Then 
\begin{equation}
T(\rho )=\sum_{k=1}^Nw_kU_k\rho \,U_k^{*},\ \ \rho \in V,  \label{eq3-is}
\end{equation}
defines a stochastic isometry $T:V\rightarrow \tilde{V}$. In the case $%
N=\infty $, the sum converges in trace norm.
\end{Prop}

\begin{proof}
The trace is a normal map and so $\text{tr}\bigl[T(\rho )\bigr]%
=\sum_{k=1}^Nw_k\text{tr}\bigl[U_k^{*}U_k\rho \bigr]=\text{tr}\bigl[\rho %
\bigr]$ for $\rho \in V$, the last equality being due to $U_k^{*}U_k=I$.
Hence $T$ is trace preserving. $T$ is positive since $U_k\rho \,U_k^{*}\in
V^{+}$ for $\rho \in V^{+}$. All $T_k:\rho \mapsto U_k\,\rho \,U_k^{*}$ are
isometric, and for $\rho \in V$, $|T_k(\rho )|\perp |T_\ell (\rho )|$ if $%
k\ne \ell $. Thus, $|T(\rho )|=\sum_{k=1}^Nw_k|T_k(\rho )|$, and 
\begin{equation*}
\left\| T(\rho )\right\| _1=\sum_{k=1}^Nw_k\left\| T_k(\rho )\right\|
_1=\sum_{k=1}^Nw_k\left\| \rho \right\| _1=\left\| \rho \right\| _1.
\end{equation*}
Convergence in trace norm follows from $\left\| \sum_{k=n}^{n+m}w_kU_k\rho
\,U_k^{*}\right\| _1\le \sum_{k=n}^{n+m}w_k\,\left\| \rho \right\| _1$ and
the fact that $\tilde{V}$ is complete.%
\end{proof}

The construction of Proposition \ref{pro3-is} turns out to be generic.

\begin{Thm}
\label{th3-is}Let $T:V\to \tilde{V}$ be a stochastic isometry. Then $\tilde{%
\mathcal{H}}$ can be decomposed as $\tilde{\mathcal{H}}=\left( \oplus
_{k=1}^N\tilde{\mathcal{H}}_k\right) \oplus \tilde{\mathcal{H}}_0$, with dim(%
$\tilde{\mathcal{H}}_k$)=dim($\mathcal{H}$), $k=1,2,\cdots ,N\le \infty $.
There exist weights $w_k>0$, $\sum_{k=1}^Nw_k=1$, and unitary or antiunitary
maps $U_k:{\mathcal{H}}\to \tilde{\mathcal{H}}_k$ such that $T$ is of the form
of Eq. (\ref{eq3-is}).
\end{Thm}

Theorem \ref{th3-is} is an immediate consequence of Propositions \ref
{pro3-dec} and \ref{pro3-mix} to be formulated next. A concept central to
the decomposition of stochastic isometries is the following.

\begin{Def}
\label{def3-mix}A stochastic map $T:V\to \tilde{V}$ will be called \textsl{($%
m$)-mixing} ($m\in \Bbb{N}$) if every pure state $P_{\left. \varphi \right. }
$, $\varphi \in {\mathcal{H}}\backslash \{0\}$, is sent to an $m$-fold
degenerate mixture; that is: $T(P_{\left. \varphi \right. })=\frac 1m\Pi
_{\left. \varphi \right. }$, where $\Pi _{\left. \varphi \right. }$ is an
orthogonal projection of rank $m$.
\end{Def}

\begin{Prop}
\label{pro3-dec}A stochastic isometry $T:V\to \tilde{V}$ can be decomposed
into a ($\sigma $-)\allowbreak {convex} combination of mixing stochastic
isometries with mutually orthogonal ranges. That is, there is a family of
mixing stochastic isometries $T_{\left. \nu \right. }:V\to \tilde{V}$, $\nu
=1,2,\cdots ,N\le \infty $ and a strictly decreasing sequence of weights $%
w_1>w_2>\cdots $, $\sum_{\nu =1}^Nw_{\left. \nu \right. }=1$, such that 
\begin{equation}
T(\rho )\ =\ \sum_{\nu =1}^Nw_{\left. \nu \right. }T_{\left. \nu \right.
}(\rho ),\ \ \ \rho \in V.  \label{eq3-dec}
\end{equation}
The sum converges in trace norm when $N=\infty $. The ranges of the $%
T_{\left. \nu \right. }$ can be represented as subspaces of ${\mathcal{B}}_1(%
\tilde{\mathcal{H}}_{\left. \nu \right. })_s$, where the $\tilde{\mathcal{H}}%
_{\left. \nu \right. }$ are mutually orthogonal closed subspaces of $\tilde{%
\mathcal{H}}$.
\end{Prop}

\begin{Prop}
\label{pro3-mix}Any $m$-mixing stochastic isometry $T:V\to \tilde{V}$ is of
the form given in Equation \ref{eq3-is}, with $N=m$, $w_k=\frac 1m$.
\end{Prop}

A positive map $T:{\mathcal{B}}_1({\mathcal{H}})_s\to {\mathcal{B}}_1
(\tilde{\mathcal{H}})_s$ is called \textsl{completely positive} if 
its canonical
extension $T^{(n)}:{\mathcal{B}}_1({\mathcal{H}})_s\otimes {\mathcal{B}}_1({\Bbb{C}}%
^n)_s\rightarrow {\mathcal{B}}_1(\tilde{\mathcal{H}})_s\otimes {\mathcal{B}}_1(%
{\Bbb{C}}^n)_s$ is positive for all $n\in \Bbb{N}$. As a specification of
Theorem \ref{th3-is} the following characterization of completely positive
stochastic isometries will be obtained.

\begin{Thm}
\label{th3-cp}A completely positive stochastic isometry $T:V\rightarrow 
\tilde{V}$ is of the form of Equation (\ref{eq3-is}), with all $U_k$
unitary, $k=1,2,\cdots ,N$.
\end{Thm}

\section{Decompositions of Stochastic Isometries\label{sec-dec}}

For $\rho \in {\mathcal{B}}_1({\mathcal{H}})$, 
$a\in {\mathcal{B}}({\mathcal{H}})$,
we will interchangeably use the notations $\langle {\rho }\,,\,{a}\rangle $
and $\langle {a}\,,\,{\rho }\rangle $ for 
$\text{tr}\bigl[\rho \cdot a\bigr]$. 
The dual $T^{*}:{\mathcal{B}}(\tilde{{\mathcal{H}}})_s\to 
{\mathcal{B}}({\mathcal{H}})_s$ of a positive linear map 
$T:{\mathcal{B}}_1({\mathcal{H}})_s\to 
{\mathcal{B}}_1(\tilde{\mathcal{H}})_s$ is a normal positive linear 
map (\cite{Dav},
Lemma 2.2.2), and any normal positive linear map on ${\mathcal{B}}(\tilde{%
{\mathcal{H}}})_s$ is the dual of a positive linear map on 
${\mathcal{B}}_1({\mathcal{H}})_s$. 
$T^{*}(I)=I$ is equivalent to $T$ being trace-preserving.
Maps $T$ on ${\mathcal{B}}_1({\mathcal{H}})_s$ and their dual maps 
$T^{*}$ will henceforth be understood to be linear without explicit mention.

\begin{Def}
\label{def4-supp}The \textsl{support} (projection) of a state $\rho \in 
\mathcal{S}$, denoted $\Pi (\rho )$, is the smallest orthogonal projection $P
$ such that $\langle {\rho }\,,\,{P}\rangle =1$.
\end{Def}

In the sequel the term projection will mean orthogonal projection. We shall
freely use properties of the complete orthocomplemented lattice of
projections. Projections $P,Q$ are called \textsl{orthogonal}, $P\perp Q$,
if $PQ=0$, or equivalently, $P+Q\le I$, or $P\le I-Q=:Q^{\perp }$. The
supremum of a family of projections $P^{(\alpha )}$, $\alpha \in A$, will be
denoted $\bigvee_{\alpha \in A}P^{(\alpha )}$. Note that the projections $%
\Pi _{\left. \varphi \right. }$ introduced in Definition \ref{def3-mix} are
the support projections of $T(P_{\left. \varphi \right. })$.

\begin{Lem}
\label{le4-is}Let $T:V\rightarrow \tilde{V}$ be a stochastic map. The
following are equivalent:\newline
{(i)}\ \ $T$ is isometric;\newline
{(ii)}\ $T^{*}\left( \Pi (T(\rho ))\right) =\Pi (\rho )$ for all $\rho \in 
\mathcal{S}$;\newline
(iii) $T^{*}\left( \Pi (T(P_{\left. \varphi \right. })\right) =P_{\left.
\varphi \right. }$ for all $\varphi \in {\mathcal{H}}\backslash \{0\}.$
\end{Lem}

\begin{proof}
Assume that $T$ is a stochastic isometry. Let $\rho \in \mathcal{S}$, then $%
0\le T^{*}\left( \Pi (T(\rho ))\right) \le I$, and $\langle {\rho }\,,\,{%
T^{*}\left( \Pi (T(\rho ))\right) }\rangle =\langle {T(\rho )}\,,\,{\Pi
(T(\rho ))}\rangle =1$. Therefore, $\Pi (\rho )\le T^{*}\left( \Pi (T(\rho
))\right) $. We show that also $T^{*}\left( \Pi (T(\rho ))\right) \le \Pi
(\rho )$, so that both operators are equal.

Write $\Pi (\rho )^{\perp }=\sum_kP_{\varphi _k}$. Then $P_{\varphi _k}\perp
\Pi (\rho )$, and so $P_{\varphi _k}\perp \rho $. Since $T$ is an
orthogonality preserving stochastic map [by Proposition \ref{pro2-orth}], it
follows that $T(P_{\varphi _k})\perp T(\rho )$. Then 
\begin{equation*}
0=\langle {T(P_{\varphi _k})}\,,\,{\Pi (T(\rho ))}\rangle =\langle {%
P_{\varphi _k}}\,,\,{T^{*}\left( \Pi (T(\rho ))\right) }\rangle .
\end{equation*}
This gives $P_{\varphi _k}\cdot T^{*}\left( \Pi (T(\rho ))\right) =0$ and
therefore one has $\Pi (\rho )^{\perp }\cdot T^{*}\left( \Pi (T(\rho
))\right) =0$. Since $T^{*}\left( \Pi (T(\rho ))\right) \le I$, it follows
that $T^{*}\left( \Pi (T(\rho ))\right) \le \Pi (\rho )$. Thus, $T^{*}\left(
\Pi (T(\rho ))\right) =\Pi (\rho )$.

Conversely, assume (ii) holds. We show that the stochastic map $T$ is
orthogonality preserving and thus, by Proposition \ref{pro2-orth}, an
isometry. Let $\rho ,\sigma \in S$, $\rho \perp \sigma $. Then $\langle {%
T(\rho )}\,,\,{\Pi (T(\sigma ))}\rangle =\langle {\rho }\,,\,{T^{*}\left(
\Pi (T(\sigma ))\right) }\rangle =\langle {\rho }\,,\,{\Pi (\sigma )}\rangle
=0$. This implies $T(\rho )\perp T(\sigma )$.

Property (iii) is entailed by (ii). The converse implication is a fairly
straightforward consequence of the fact that $T^{*}$ is normal.%
\end{proof}

\begin{Lem}
\label{le4-dec1}Let $T:V\rightarrow \tilde{V}$ be a stochastic isometry and
not pure. There exists a strictly decreasing sequence $\tilde{w}_{\left. \nu
\right. }$, $0<\tilde{w}_{\left. \nu \right. }<1$, $\nu =1,2,\cdots ,N\le
\infty $, and a sequence of numbers $m_{\left. \nu \right. }\in \Bbb{N}$
such that $\sum_{ \nu }\tilde{w}_{\left. \nu \right. }m_{\left.
\nu \right. }=1$ and 
\begin{equation}
T(P_{\left. \varphi \right. })=\sum_{\nu =1}^N\tilde{w}_{\left. \nu \right.
}\Pi _{\left. \varphi \right. }^{\left. \nu \right. }  \label{eq4-dec1}
\end{equation}
for all $\varphi \in {\mathcal{H}}\backslash \{0\}$. Here the $\Pi _{\left.
\varphi \right. }^{\left. \nu \right. }$ are projections with rank $%
m_{\left. \nu \right. }$. Furthermore, $\Pi _{\left. \varphi \right.
}^{\left. \nu \right. }\perp \Pi _{\left. \psi \right. }^{\left. \mu \right.
}$ for all $\varphi ,\psi \in {\mathcal{H}}\backslash \{0\}$ if $\nu \ne \mu $.
\end{Lem}

\begin{proof}
For $\varphi ,\psi \ne 0$, write the spectral decompositions of $T(P_{\left.
\varphi \right. })$, $T(P_{\left. \psi \right. })$ as 
\begin{equation}
T(P_{\left. \varphi \right. })=\sum_{\nu =1}^Na_{\left. \nu \right.
}P_{\varphi _{\left. \nu \right. }},\ \ T(P_{\left. \psi \right.
})=\sum_{\mu =1}^Mb_{\left. \mu \right. }P_{\psi _{\left. \mu \right. }},
\label{eq4-dec2}
\end{equation}
respectively, where the $\varphi _{\left. \nu \right. }$ and $\psi _{\left.
\mu \right. }$ form orthonormal systems, $N,M\in \Bbb{N\cup \{\infty \}}$,
and the $a_{\left. \nu \right. },b_{\left. \mu \right. }$ are the nonzero
eigenvalues. We show that the eigenvalues coincide in numerical values and
multiplicities. First observe that $P_{\varphi _{\left. \nu \right. }}\le
\Pi (T(P_{\left. \varphi \right. }))$, so that, by virtue of Lemma \ref
{le4-dec1}, $T^{*}\big(P_{\varphi _{\left. \nu \right. }}\big)\le P_{\left.
\varphi \right. }$, $\langle {P}${$_{\left. \varphi \right. }$}$\,,\,{%
T^{*}(P_{\varphi _{\left. \nu \right. }})}\rangle =\langle {T(P}${$_{\left.
\varphi \right. }$}${)}\,,\,{P_{\varphi _{\left. \nu \right. }}}\rangle
=a_{\left. \nu \right. }$, and therefore 
\begin{equation}
T^{*}(P_{\varphi _{\left. \nu \right. }})\ =\ a_{\left. \nu \right.
}P_{\left. \varphi \right. }.  \label{eq4-dec3}
\end{equation}
Similarly, 
\begin{equation}
T^{*}(P_{\psi _{\left. \mu \right. }})\ =\ b_{\left. \mu \right. }P_{\left.
\psi \right. }.  \label{eq4-dec4}
\end{equation}
Case 1: $0\ne |\left\langle \,\varphi \,|\,\psi \,\right\rangle |(\ne 1)$.
Take unit vectors $\varphi ^{\prime },\psi ^{\prime }$ in the span of $%
\varphi ,\psi $ such that $\left\langle \,\varphi \,|\,\varphi ^{\prime
}\,\right\rangle =\left\langle \,\psi \,|\,\psi ^{\prime }\,\right\rangle =0$%
. Then $P_{\left. \varphi \right. }+P_{\varphi ^{\prime }}=P_{\left. \psi
\right. }+P_{\psi ^{\prime }}$ and therefore $T(P_{\left. \varphi \right.
})+T(P_{\varphi ^{\prime }})=T(P_{\left. \psi \right. })+T(P_{\psi ^{\prime
}})$. Since $T$ is orthogonality preserving, this yields: 
\begin{equation*}
\big(T(P_{\left. \varphi \right. })\big)^2=T(P_{\left. \psi \right.
})T(P_{\left. \varphi \right. })+T(P_{\psi ^{\prime }})T(P_{\left. \varphi
\right. }),\ \ 
\end{equation*}
and so 
\begin{equation*}
T(P_{\left. \psi \right. })\,\big(T(P_{\left. \varphi \right. })\big)^2=\big(%
T(P_{\left. \psi \right. })\big)^2T(P_{\left. \varphi \right. }),
\end{equation*}
and hence by induction 
\begin{equation}
\big(T(P_{\left. \psi \right. })\big)^n\big(T(P_{\left. \varphi \right. })%
\big)^m=\big(T(P_{\left. \psi \right. })\big)^m\big(T(P_{\left. \varphi
\right. })\big)^n,\ \ n,m\in \Bbb{N}.  \label{eq4-prod}
\end{equation}
On multiplication of Eq. (\ref{eq4-prod}) with $P_{\psi _{\left. \mu \right.
}}$ from the left and with $P_{\varphi _{\left. \nu \right. }}$ from the
right, we obtain 
\begin{equation}
b_{\left. \mu \right. }^na_{\left. \nu \right. }^m\,P_{\psi _{\left. \mu
\right. }}P_{\varphi _{\left. \nu \right. }}=b_{\left. \mu \right.
}^ma_{\left. \nu \right. }^n\,P_{\psi _{\left. \mu \right. }}P_{\varphi
_{\left. \nu \right. }}.  \label{eq4-ab1}
\end{equation}

Choosing $n=2,m=1$, one concludes that 
\begin{equation}
a_{\left. \nu \right. }=b_{\left. \mu \right. }\ \ \text{whenever}\ \
P_{\psi _{\left. \mu \right. }}P_{\varphi _{\left. \nu \right. }}\ne 0.
\label{eq4-ab2}
\end{equation}
Applying (\ref{eq4-prod}) again, one can write 
\begin{equation*}
\langle {\big(T(P_{\left. \psi \right. })\big)^{n-1}}\,,\,{T(P_{\left.
\varphi \right. })}\rangle =\langle {T(P_{\left. \psi \right. })}\,,\,{\big(%
T(P_{\left. \varphi \right. })\big)^{n-1}}\rangle ,\ n\ge 2,
\end{equation*}
that is, by virtue of (\ref{eq4-dec2}), 
\begin{equation*}
\sum_{\left. \mu \right. }b_{\left. \mu \right. }^{n-1}\langle {P_{\psi
_{\left. \mu \right. }}}\,,\,{T(P_{\left. \varphi \right. })}\rangle
=\sum_{\left. \nu \right. }a_{\left. \nu \right. }^{n-1}\langle {T(P_{\left.
\psi \right. })}\,,\,{P_{\varphi _{\left. \nu \right. }}}\rangle .
\end{equation*}
Observing that [by virtue of (\ref{eq4-dec3}),(\ref{eq4-dec4})] 
\begin{eqnarray}
{\langle {P_{\psi _{\left. \mu \right. }}}\,,\,{T(P_{\left. \varphi \right.
})}\rangle } &=&{\langle {T^{*}(P_{\psi _{\left. \mu \right. }})}\,,\,{%
P_{\left. \varphi \right. }}\rangle =b_{\left. \mu \right. }\langle {%
P_{\left. \psi \right. }}\,,\,{P_{\left. \varphi \right. }}\rangle ,}
\label{eq4-ab3} \\
{\langle {T(P_{\left. \psi \right. })}\,,\,{P_{\varphi _{\left. \nu \right.
}}}\rangle } &=&{\langle {P_{\left. \psi \right. }}\,,\,{T^{*}(P_{\varphi
_{\left. \nu \right. }})}\rangle =a_\nu \langle {P_{\left. \psi \right. }}%
\,,\,{P_{\left. \varphi \right. }}\rangle ,}  \label{eq4-ab4}
\end{eqnarray}
and using the fact that $\langle {P}${${_{\left. \psi \right. }}$}$\,,\,{P}${%
$_{\left. \varphi \right. }$}$\rangle \ne 0$, we obtain 
\begin{equation}
\sum_{\left. \nu \right. }a_{\left. \nu \right. }^n=\sum_{\left. \mu \right.
}b_{\left. \mu \right. }^n,\ \ \text{for all}\ \ n\in \Bbb{N}.
\label{eq4-ab5}
\end{equation}
We rewrite this, making the multiplicities explicit: thus we let $%
a_1,a_2,\cdots $ and $b_1,b_2,\cdots $ denote the strictly decreasing
sequences of eigenvalues of $T(P_{\left. \varphi \right. })$, $T(P{{_{\left.
\psi \right. }}})$, with multiplicities $n_{\left. \nu \right. },m{_{\left.
\mu \right. }}$, respectively. Then Eq. (\ref{eq4-ab5}) reads: 
\begin{equation}
\sum_{\left. \nu \right. }n_{\left. \nu \right. }a_{\left. \nu \right.
}^n=\sum_{\left. \mu \right. }m{_{\left. \mu \right. }}b_{\left. \mu \right.
}^n,\ \ \text{for all}\ \ n\in \Bbb{N}.  \label{eq4-ab6}
\end{equation}
The orthonormal systems of eigenvectors shall now be denoted $\varphi _{\nu
,i}$, $\psi _{\mu ,j}$, where $i\in \{1,\cdots ,n_{\left. \nu \right. }\}$, $%
j\in \{1,\cdots ,m_{\left. \mu \right. }\}$. Now observe that due to (\ref
{eq4-ab3}),(\ref{eq4-ab4}), 
\begin{equation*}
\langle {T(P{_{\left. \psi \right. }})}\,,\,{P_{\varphi _{\left. \nu
,i\right. }}}\rangle =\langle {P{_{\left. \psi \right. }}}\,,\,{%
T^{*}(P_{\varphi _{\left. \nu ,i\right. }})}\rangle =a_{\left. \nu \right.
}\langle {P{_{\left. \psi \right. }}}\,,\,{P_{\left. \varphi \right. }}%
\rangle ,
\end{equation*}
which implies that for each $\nu $ there must be $\mu ,j$ such that $\langle 
{P_{\psi _{\mu ,j}}}\,,\,{P_{\varphi _{\nu ,i}}}\rangle \ne 0$. Therefore,
by (\ref{eq4-ab2}), $a_{\left. \nu \right. }=b_{\left. \mu \right. }$. A
similar reasoning entails that for each $\mu $ there must be $\nu $ such
that $a_{\left. \nu \right. }=b_{\left. \mu \right. }$. Since the two
sequences of eigenvalues are strictly decreasing, they must be identical: $%
a_{\left. \nu \right. }=b_{\left. \nu \right. }$, for all values of $\nu $.
It remains to be shown that the multiplicities coincide as well. Equation (%
\ref{eq4-ab6}) can be written as 
\begin{equation}
\sum_{\nu =1}^N(n_{\left. \nu \right. }-m_{\left. \nu \right. })\,a_{\left.
\nu \right. }^n\ =\ 0\ \ \text{for all}\ \ n\in \Bbb{N}.  \label{eq4-mult}
\end{equation}
We show that $x_{\left. \nu \right. }:=n_{\left. \nu \right. }-m_{\left. \nu
\right. }=0$ for all $\nu =1,\cdots ,N$. This is obvious for $N=1$. So let $%
N>1$. Suppose $x_1\ne 0$. Then (\ref{eq4-mult}) gives 
\begin{equation*}
0=1+\left( a_1^nx_1\right) ^{-1}\,\sum_{\kappa =2}^Na_{\left. \kappa \right.
}^nx_{\left. \kappa \right. },
\end{equation*}
that is, 
\begin{equation*}
\sum_{\kappa =2}^N\left( \frac{a_{\left. \kappa \right. }}{a_1}\right) ^n%
\frac{x_{\left. \kappa \right. }}{x_1}=-1.
\end{equation*}
But then 
\begin{equation*}
1=\left| \sum_{\kappa =2}^N\left( \frac{a_{\left. \kappa \right. }}{a_1}%
\right) ^n\frac{x_{\left. \kappa \right. }}{x_1}\right| \le \sum_{\kappa
=2}^N\left( \frac{a_{\left. \kappa \right. }}{a_1}\right) ^n\left| \frac{%
x_{\left. \kappa \right. }}{x_1}\right| =:S(n)\ (<\infty ).
\end{equation*}
Since $a_1>a_2>\cdots $, we also have: 
\begin{equation*}
S(n+1)=\sum_{\kappa =2}^N\left( \frac{a_{\left. \kappa \right. }}{a_1}%
\right) ^{n+1}\left| \frac{x_{\left. \kappa \right. }}{x_1}\right| \le
\left( \frac{a_2}{a_1}\right) ^nS(1).
\end{equation*}
Hence $S(n)\downarrow 0$ and $S(n)\ge 1$, which is a contradiction.
Therefore, $x_1=0$. The argument can be repeated for $x_2,x_3,\cdots $ to
yield $x_{\left. \nu \right. }=0$ for all $\nu $.\newline
Note that Eq.\ (\ref{eq4-ab2}) also implies 
\begin{equation}
\langle {P_{\psi _{\mu ,j}}}\,,\,{P_{\varphi _{\nu ,i}}}\rangle =0\ \ \text{%
for}\ \ \nu \ne \mu .  \label{eq4-orth}
\end{equation}
Finally we get, putting $a_{\left. \nu \right. }=\tilde{w}_{\left. \nu
\right. }$: 
\begin{equation}
T(P_{\left. \varphi \right. })=\sum_{\nu =1}^N\tilde{w}_{\left. \nu \right.
}\left( \sum_{i=1}^{m_{\left. \nu \right. }}P_{\varphi _{\nu ,i}}\right) ,\
\ T(P_{\left. \psi \right. })=\sum_{\nu =1}^N\tilde{w}_{\left. \nu \right.
}\left( \sum_{j=1}^{m_{\left. \nu \right. }}P_{\psi _{\left. \nu ,j\right.
}}\right) ,  \label{eq4-dec5}
\end{equation}
where $\sum_{\left. \nu \right. }\tilde{w}_{\left. \nu \right. }n_{\left.
\nu \right. }=1$. This is of the form of Eq.\ (4.1), and due to (\ref
{eq4-orth}) the projections $\Pi _{\left. \varphi \right. }^{\left. \nu
\right. }$, $\Pi _{\left. \psi \right. }^{\left. \mu \right. }$ are mutually
orthogonal for any $\varphi ,\psi \ne 0$ ($\langle \varphi |\psi \rangle \ne
0$) if $\nu \ne \mu $.

Case 2: $\left\langle \,\varphi \,|\,\psi \,\right\rangle =0$. Replace $\psi 
$ by a unit vector $\tilde{\psi}$ in the span of $\varphi ,\psi $ and not
parallel to $\varphi $ or $\psi $. Then apply the argument of Case 1 to
obtain equations (\ref{eq4-dec5}) for the pairs $\varphi ,\tilde{\psi}$ and $%
\tilde{\psi},\psi $. It follows that equations of this form also hold for
orthogonal pairs $\varphi ,\psi $.%
\end{proof}

\begin{Lem}
\label{le4-dec2}Let $T:V\rightarrow \tilde{V}$ be a stochastic isometry.
There exists a complete family of mutually orthogonal projections $P_{\left.
\nu \right. }$, $\nu =0,1,\cdots ,N\le \infty $, $\sum_{\nu =0}^NP_{\left.
\nu \right. }=I$, such that 
\begin{equation}
T(\rho )=\sum_{\nu =0}^NP_{\left. \nu \right. }\,T(\rho )\,P_{\left. \nu
\right. }.  \label{eq4-dec6}
\end{equation}
For $\nu =1,\cdots ,N$, $\varphi \in {\mathcal{H}}\backslash \{0\}$, let $\Pi
_{\left. \varphi \right. }^{\left. \nu \right. }$ be the projections of rank 
$m_{\left. \nu \right. }$ established in Lemma \ref{le4-dec1}. Then $%
P_{\left. \nu \right. }=\bigvee_{\varphi \in {\mathcal{H}}\backslash \{0\}}\Pi
_{\left. \varphi \right. }^{\left. \nu \right. }$ for $\nu =1,\cdots N$, and 
$P_0=I-\sum_{\nu =1}^NP_{\left. \nu \right. }$.
\end{Lem}

\begin{proof}
The projections $P_{\left. \nu \right. }$ are mutually orthogonal by virtue
of the orthogonality of all $\Pi _{\left. \varphi \right. }^{\left. \nu
\right. }$, $\Pi _{\left. \psi \right. }^{\left. \mu \right. }$, $\nu \ne
\mu $. Then $P_{\left. \nu \right. }\Pi _{\left. \varphi \right. }^{\left.
\mu \right. }=\delta _{\mu \nu }\Pi _{\left. \varphi \right. }^{\left. \nu
\right. }=\Pi _{\left. \varphi \right. }^{\left. \mu \right. }P_{\left. \nu
\right. }$, and so $\sum_{\left. \nu \right. }P_{\left. \nu \right.
}T(P_{\left. \varphi \right. })\,P_{\left. \nu \right. }=\sum_{\left. \nu
\right. }\tilde{w}_{\left. \nu \right. }\Pi _{\left. \varphi \right.
}^{\left. \nu \right. }=T(P_{\left. \varphi \right. })$. By continuity of $T$%
, this equality extends to all $T(\rho )$, $\rho \in V$.%
\end{proof}

\noindent \textbf{Proof of Proposition}\textsc{\ }\textbf{\ref{pro3-dec}.}
For the $\tilde{w}_{\left. \nu \right. },m_{\left. \nu \right. }$ given in
Lemma \ref{le4-dec1}, put $w_{\left. \nu \right. }=\tilde{w}_{\left. \nu
\right. }m_{\left. \nu \right. }$, and define the maps $T_{\left. \nu
\right. }$ via 
\begin{equation*}
T_{\left. \nu \right. }(P_{\left. \varphi \right. }):=\frac 1{m_{\left. \nu
\right. }}\Pi _{\left. \varphi \right. }^{\left. \nu \right. }={w_{\left.
\nu \right. }}^{-1}P{_{\left. \nu \right. }}T(P_{\left. \varphi \right. })P{%
_{\left. \nu \right. }},
\end{equation*}
with the projections $P{_{\left. \nu \right. }},\nu =1,2,\cdots ,N$ given in
Lemma \ref{le4-dec2}. Then (\ref{eq4-dec1}) realizes (\ref{eq3-dec}), and in
view of (\ref{eq4-dec6}), the maps $T{_{\left. \nu \right. }}$ are mixing
stochastic maps on $V$ with ranges in ${\mathcal{H}}_1(\tilde{{\mathcal{H}}}{%
_{\left. \nu \right. }})_s$, $\tilde{{\mathcal{H}}}{_{\left. \nu \right. }}=P{%
_{\left. \nu \right. }}\tilde{{\mathcal{H}}}$.

\section{Mixing Stochastic Isometries\label{sec-mix}}

We establish some properties of the $m$-mixing stochastic isometries which
are instrumental in proving Proposition \ref{pro3-mix}. The proof given here
applies Hilbert space techniques and emphasizes elementary geometrical
aspects, particularly the preservation of orthogonality. In the appendix an
alternative proof is presented that is based on a link with Kadison's
characterization of isometries of operator algebras.

\begin{Lem}
\label{le5-mix1}Let $T:V\to \tilde{V}$ be an $m$-mixing stochastic isometry,
with $T(P_\phi )=\frac 1m\Pi _\phi $, $\phi \in {\mathcal{H}}\backslash \{0\}$%
. Let $\varphi ,\psi \in {\mathcal{H}},\xi \in \tilde{{\mathcal{H}}}$ be unit
vectors. Then the following hold: 
\begin{equation}
P_{\left. \xi \right. }\le \Pi _{\left. \varphi \right. }\ \ \
\Longrightarrow \ \ T^{*}(P_{\left. \xi \right. })\ =\ \frac 1mP_{\left.
\varphi \right. };  \label{eq5-mix1}
\end{equation}
\begin{equation}
P_{\left. \xi \right. }\le \Pi _{\left. \varphi \right. }\ \ \
\Longrightarrow \ \ \langle {P_{\left. \xi \right. }}\,,\,{\Pi _{\left. \psi
\right. }}\rangle \ =\ \langle {P_{\left. \varphi \right. }}\,,\,{P_{\left.
\psi \right. }}\rangle ;  \label{eq5-mix2}
\end{equation}
\begin{equation}
\Pi _{\left. \varphi \right. }\Pi {_{\left. \psi \right. }}\Pi _{\left.
\varphi \right. }\ =\ \langle {P_{\left. \varphi \right. }}\,,\,{P_{\left.
\psi \right. }}\rangle \,\Pi _{\left. \varphi \right. }.  \label{eq5-mix3}
\end{equation}
\end{Lem}

\begin{proof}
Let $P{_{\left. \xi \right. }}\le \Pi _{\left. \varphi \right.
}=mT(P_{\left. \varphi \right. })$. Then by Lemma \ref{le4-is}, $0\le T^{*}(P%
{_{\left. \xi \right. }})\le T^{*}(\Pi _{\left. \varphi \right. })=P_{\left.
\varphi \right. }$, and therefore $T^{*}(P{_{\left. \xi \right. }}%
)=a\,P_{\left. \varphi \right. }$, $0\le a\le 1$. Further, $1=\langle {P}${$%
_{\left. \xi \right. }$}$\,,\,{\Pi }${$_{\left. \varphi \right. }$}$\rangle
=\langle {T^{*}(P}${$_{\left. \xi \right. }$}${)}\,,\,{m\,P}${$_{\left.
\varphi \right. }$}$\rangle =a\cdot m$, hence $a=\frac 1m$. This proves (\ref
{eq5-mix1}).

Let $P{_{\left. \xi \right. }}\le \Pi _{\left. \varphi \right. }$, then with 
$\Pi {_{\left. \psi \right. }}=mT(P{_{\left. \psi \right. }})$ and (\ref
{eq5-mix1}) one obtains 
\begin{equation*}
\langle {P_{\left. \xi \right. }}\,,\,{\Pi _{\left. \psi \right. }}\rangle
=m\langle {T^{*}(P_{\left. \xi \right. })}\,,\,{P_{\left. \psi \right. }}%
\rangle =\langle {P_{\left. \varphi \right. }}\,,\,{P_{\left. \psi \right. }}%
\rangle ,
\end{equation*}
that is, (\ref{eq5-mix2}).

Let $\{\varphi _k|k=1,\cdots m\}\cup \{\varphi _\ell ^{\prime }|\ell \in L\}$
be an orthonormal basis of ${\mathcal{H}}$ such that $\Pi _{\left. \varphi
\right. }=\sum_{k=1}^mP_{\varphi _k}$ and thus $\Pi _{\left. \varphi \right.
}\varphi _\ell ^{\prime }=0$ for all $\ell $. Then the only nonvanishing
matrix elements of the (finite rank) operator $\Pi _{\left. \varphi \right.
}\Pi {_{\left. \psi \right. }}P_{\left. \varphi \right. }$ are those
obtained from the $\varphi _k$, and (\ref{eq5-mix3}) is equivalent to 
\begin{equation}
\left\langle \,\varphi _j\,|\,\Pi {_{\left. \psi \right. }}\varphi
_k\,\right\rangle =\langle {P_{\left. \varphi \right. }}\,,\,{P_{\left. \psi
\right. }}\rangle \,\delta _{jk}.  \label{eq5-mix4}
\end{equation}
For $j=k$ one has $\left\langle \,\varphi _k\,|\,\Pi {_{\left. \psi \right. }%
}\varphi _k\,\right\rangle =\langle {P_{\varphi _k}}\,,\,{\Pi }${$_{\left.
\psi \right. }$}$\rangle $, which equals $\langle {P}${$_{\left. \varphi
\right. }$}$\,,\,{P}${$_{\left. \psi \right. }$}$\rangle $ by virtue of
(5.2).

For $j\ne k$, consider the unit vectors $\eta :=\left( \varphi _j+\alpha
\varphi _k\right) /\surd 2$, $|\alpha |=1$. We have $P_{\left. \eta \right.
}\le \Pi _{\left. \varphi \right. }$ and therefore, by (\ref{eq5-mix2}), 
\begin{equation}
\langle {P_{\left. \eta \right. }}\,,\,{\Pi _{\left. \psi \right. }}\rangle
=\langle {P_{\left. \varphi \right. }}\,,\,{P_{\left. \psi \right. }}\rangle
.  \label{eq5-mix5}
\end{equation}

Now observe that for any unit vectors $\varphi ,\psi $, $a\varphi +b\psi $ ($%
a,b\in \Bbb{C}$), the projection $P_{a\varphi +b\psi }$ can be written as 
\begin{equation}
P_{a\varphi +b\psi }\ =\ |a|^2P_{\left. \varphi \right. }+|b|^2P{_{\left.
\psi \right. }}+a\overline{b}\,A_{\varphi \psi }+\overline{a}b\,A_{\psi
\varphi },  \label{eq5-pol1}
\end{equation}
where $A_{\varphi \psi }$is the operator of rank one defined via 
\begin{equation}
A_{\varphi \psi }\xi \ =\ \varphi \,\left\langle \,\psi \,|\,\xi
\,\right\rangle ,\ \ \xi \in {\mathcal{H}}.  \label{eq5-pol2}
\end{equation}
Thus we obtain 
\begin{equation}
P{_{\left. \eta \right. }}\ =\ P_{(\varphi _j+a\varphi _k)/\surd 2}\ =\
\tfrac 12P_{\varphi _j}\,+\tfrac 12P_{\varphi _k}\,+\,\tfrac
12\{a\,A_{\varphi _k\varphi _j}\,+\,\overline{a}\,A_{\varphi _j\varphi _k}\}.
\label{eq5-pol3}
\end{equation}

\noindent Again by (\ref{eq5-mix2}), we have $\langle {P_{\varphi _j}}\,,\,{%
\Pi }${$_{\left. \psi \right. }$}$\rangle =\langle {P_{\varphi _k}}\,,\,{\Pi 
}${$_{\left. \psi \right. }$}$\rangle \allowbreak =\langle {P}${$_{\left.
\varphi \right. }$}$\,,\,{P}${$_{\left. \psi \right. }$}$\rangle $. Then
combining (\ref{eq5-pol3}) and (\ref{eq5-mix5}) yields 
\begin{equation*}
\text{Re}\left\{ a\,\left\langle \,\varphi _j\,|\,\Pi _{\left. \psi \right.
}\varphi _k\,\right\rangle \right\} =0.
\end{equation*}
Choosing for $a$ the values $a=1,i$, one concludes that $\left\langle
\,\varphi _j\,|\,\Pi _{\left. \psi \right. }\varphi _k\,\right\rangle =0$.
Hence (\ref{eq5-mix4}) follows, and so (\ref{eq5-mix3}) is verified.%
\end{proof}

\begin{Prop}
\label{pro5-mix1}Let $T:V\rightarrow \tilde{V}$ be an $m$-mixing stochastic
map. The following statements are equivalent.\newline
\noindent (i)\ \ $T$ is a stochastic isometry;\newline
\noindent (ii)\ $P{_{\left. \xi \right. }}\le \Pi _{\left. \varphi \right.
}\ \ \ \Longrightarrow \ \ T^{*}(P{_{\left. \xi \right. }})\ =\ \frac
1mP_{\left. \varphi \right. }$ for all $\varphi \in {\mathcal{H}}\backslash
\{0\},\xi \in \tilde{{\mathcal{H}}}\backslash \{0\}$;\newline
\noindent (iii) $P{_{\left. \xi \right. }}\le \Pi _{\left. \varphi \right.
}\ \ \ \Longrightarrow \ \ \langle {P}${$_{\left. \xi \right. }$}$\,,\,{\Pi }
${$_{\left. \psi \right. }$}$\rangle \ =\ \langle {P}${$_{\left. \varphi
\right. }$}$\,,\,{P}${$_{\left. \psi \right. }$}$\rangle $ for all $\xi
,\varphi ,\psi \in {\mathcal{H}}\backslash \{0\},\xi \in \tilde{{\mathcal{H}}}%
\backslash \{0\}$;

\noindent (iv)$\ \Pi _{\left. \varphi \right. }\Pi _{\left. \psi \right.
}\Pi _{\left. \varphi \right. }\ =\ \langle {P}${$_{\left. \varphi \right. }$%
}$\,,\,{P}${$_{\left. \psi \right. }$}$\rangle \,\Pi _{\left. \varphi
\right. }$ for all $\varphi ,\psi \in {\mathcal{H}}\backslash \{0\}$.
\end{Prop}

\begin{proof}
According to Lemma \ref{le5-mix1}, (i) implies each of the statements (ii),
(iii), (iv). We show that each of the latter statements implies that $T$ is
orthogonality preserving for pairs of pure states, so that, by Proposition 
\ref{pro2-orth}, $T$ is a stochastic isometry.

Assume (ii) holds. Let $P_{\left. \varphi \right. }\perp P_{\left. \psi
\right. }$. Take $\xi \in \tilde{{\mathcal{H}}}\backslash \{0\}$ such that $P{%
_{\left. \xi \right. }}\le \Pi _{\left. \varphi \right. }$. Then 
\begin{equation*}
\langle {P_{\left. \xi \right. }}\,,\,{\Pi _{\left. \psi \right. }}\rangle
=\langle {P_{\left. \xi \right. }}\,,\,{mT(P_{\left. \psi \right. })}\rangle
=\langle {mT^{*}(P_{\left. \xi \right. })}\,,\,{P_{\left. \psi \right. }}%
\rangle =\langle {P_{\left. \varphi \right. }}\,,\,{P_{\left. \psi \right. }}%
\rangle =0.
\end{equation*}
Hence $P{_{\left. \xi \right. }}\perp \Pi {_{\left. \psi \right. }}$. This
holds for any $P{_{\left. \xi \right. }}\le \Pi _{\left. \varphi \right. }$,
and therefore $\Pi _{\left. \varphi \right. }\perp \Pi {_{\left. \psi
\right. }}$.

Assume (iii) holds. Let $P_{\left. \varphi \right. }\perp P{_{\left. \psi
\right. }}$. It follows that for any $\xi \in \tilde{{\mathcal{H}}}\backslash
\{0\}$ with $P{_{\left. \xi \right. }}\le \Pi _{\left. \varphi \right. }$, $%
\langle {P}${$_{\left. \xi \right. }$}$\,,\,{\Pi }${$_{\left. \psi \right. }$%
}$\rangle =\langle {P}${$_{\left. \varphi \right. }$}$\,,\,{P}${$_{\left.
\psi \right. }$}$\rangle =0$, and therefore $\Pi _{\left. \varphi \right.
}\perp \Pi {_{\left. \psi \right. }}$.

Assume (iv) holds. Let $P_{\left. \varphi \right. }\perp P{_{\left. \psi
\right. }}$. It follows that $\Pi _{\left. \varphi \right. }\Pi {_{\left.
\psi \right. }}\Pi _{\left. \varphi \right. }=\langle {P}${$_{\left. \varphi
\right. }$}$\,,\,{P}${$_{\left. \psi \right. }$}$\rangle \Pi _{\left.
\varphi \right. }=0$, and so $\Pi _{\left. \varphi \right. }\perp \Pi {%
_{\left. \psi \right. }}$.%
\end{proof}

The following generalization of the relation (\ref{eq5-mix3}) will be
crucial.

\begin{Prop}
\label{pro5-mix2}Let $T$ be an $m$-mixing stochastic isometry. For unit
vectors $\varphi ,\psi ,\vartheta \in {\mathcal{H}}$, the following relation
holds between the projections $\Pi _{\left. \varphi \right. }=mT(P_{\left.
\varphi \right. })$, $\Pi {_{\left. \psi \right. }}=mT(P{_{\left. \psi
\right. }})$, $\Pi _{\left. \vartheta \right. }=mT(P_{\left. \vartheta
\right. })$: 
\begin{eqnarray}
\Pi _{\left. \varphi \right. }\Pi _{\left. \vartheta \right. }\Pi {_{\left.
\psi \right. }}\Pi _{\left. \varphi \right. }=a\,\,\Pi _{\left. \varphi
\right. }^{(L)}\,+\overline{a}\,\,\Pi _{\left. \varphi \right. }^{(A)},
\label{eq5-prod1} \\
a=\left\langle \,\varphi \,|\,\vartheta \,\right\rangle \,\left\langle
\,\vartheta \,|\,\psi \,\right\rangle \,\left\langle \,\psi \,|\,\varphi
\,\right\rangle .  \notag
\end{eqnarray}
Here $\Pi _{\left. \varphi \right. }^{(L)},\Pi _{\left. \varphi \right.
}^{(A)}$ are projections independent of $\psi ,\vartheta $ and determined
solely by $P_{\left. \varphi \right. }$. They satisfy $\Pi _{\left. \varphi
\right. }^{(L)}+\Pi _{\left. \varphi \right. }^{(A)}=\Pi _{\left. \varphi
\right. }$.
\end{Prop}

\noindent The proof will be based the study of a number of special cases and
on exploiting the linearity of $T$. We note two trivial special cases: if
any single pair among the three vectors are mutually orthogonal then the
left hand side of Eq.\ (\ref{eq5-prod1}) is identically 0; and if any two of
these vectors are linearly dependent then (\ref{eq5-prod1}) reduces to (\ref
{eq5-mix3}).

\begin{Lem}
\label{le5-mix2}Let $T$ be an $m$-mixing stochastic isometry. For any pair
of mutually orthogonal unit vectors $\varphi _1,\varphi _2\in {\mathcal{H}}$
the following holds: 
\begin{equation}
\Pi _{\varphi _1}\,\Pi _{\varphi _1+i\varphi _2}\Pi _{\varphi _1+\varphi
_2}\Pi _{\varphi _1}\ =\ {{\textstyle{\frac{1+i}4}}}\,\Pi _{\varphi
_1}^{(L)}\,+\,{{\textstyle{\frac{1-i}4}}}\,\Pi _{\varphi _1}^{(A)}.
\label{eq5-prod2}
\end{equation}
Here $\Pi _{\varphi _1}^{(L)},\Pi _{\varphi _1}^{(A)}$ are projections
satisfying $\Pi _{\varphi _1}^{(L)}+\Pi _{\varphi _1}^{(A)}=\Pi _{\varphi _1}
$.
\end{Lem}

\begin{proof}
We note first that at this stage we do not claim the independence of the
projections $\Pi _{\varphi _1}^{(L)},\Pi _{\varphi _1}^{(A)}$ of the choice
of $\varphi _2$. This will be established in a later step.

We will frequently use Eqs. (\ref{eq5-mix2}) and (\ref{eq5-mix3}) without
explicit mentioning. Let $\{\varphi _{1k}\,|\,k=1,\cdots ,m\}$ be an
arbitrary orthonormal basis (\textsc{onb}) of the subspace $\Pi _{\varphi _1}%
{\mathcal{H}}$. Then, by virtue of Eq. (\ref{eq5-mix3}), for $\alpha \in \Bbb{C%
}$, $|\alpha |=1$, the system of vectors $\varphi _{2k}(\alpha )=\allowbreak
\allowbreak 2\Pi _{\varphi _2}\Pi _{\varphi _1+\alpha \varphi _2}\varphi
_{1k}$, $k=1,\cdots ,m$, form an \textsc{onb} of $\Pi _{\varphi _2}\mathcal{H%
}$. Furthermore, 
\begin{eqnarray}
{\Pi _{\varphi _1+\alpha \varphi _2}}\left( {\varphi _{1k}+\varphi
_{2k}(\alpha )}\right) &=&{\Pi _{\varphi _1+\alpha \varphi _2}}\left( {%
\varphi _{1k}}\right) {+2\Pi _{\varphi _1+\alpha \varphi _2}\Pi _{\varphi
_2}\Pi _{\varphi _1+\alpha \varphi _2}}\left( {\varphi _{1k}}\right)  \notag
\\
&=&{2\Pi _{\varphi _1+\alpha \varphi _2}}\left( {\varphi _{1k}}\right) 
\notag \\
&=&2\left( \Pi _{\varphi _1}+\Pi _{\varphi _2}\right) {\Pi _{\varphi
_1+\alpha \varphi _2}\Pi _{\varphi _1}}\left( {\varphi _{1k}}\right)
\label{eq5-eig} \\
&=&{\varphi _{1k}+\varphi _{2k}(\alpha )}  \notag
\end{eqnarray}
In the third line we have used the fact that $\Pi _{\varphi _1+\alpha
\varphi _2}\perp \Pi _{\varphi _1-\alpha \varphi _2}$ (since $T$ is
orthogonality preserving) and that, by virtue of the linearity of $T$, 
\begin{equation*}
\Pi _{\varphi _1+\alpha \varphi _2}+\Pi _{\varphi _1-\alpha \varphi _2}=\Pi
_{\varphi _1}+\Pi _{\varphi _2},
\end{equation*}
and therefore 
\begin{equation}
\Pi _{\varphi _1+\alpha \varphi _2}\le \Pi _{\varphi _1}+\Pi _{\varphi _2}.
\label{eq5-prod3}
\end{equation}
Let us denote $\varphi _{2k}(1)=\varphi _{2k}$ and $\varphi _{2k}(i)=\varphi
_{2k}^{\prime }$. Then (\ref{eq5-eig}) entails $P_{\varphi _{1k}+\varphi
_{2k}}\le \Pi _{\varphi _1+\varphi _2}$, and so [by (\ref{eq5-mix2})] 
\begin{equation}
\langle {P_{\varphi _{1k}+\varphi _{2k}}}\,,\,{\Pi _{\varphi _1+i\varphi _2}}%
\rangle =\langle {P_{\varphi _1+\varphi _2}}\,,\,{P_{\varphi _1+i\varphi _2}}%
\rangle ={{\textstyle{\frac 14}}}|1+i|^2={{\textstyle{\frac 12}}}.
\label{eq5-prod4}
\end{equation}
But from (\ref{eq5-pol1}) we have 
\begin{equation*}
P_{\varphi _{1k}+\varphi _{2k}}={{\textstyle{\frac 12}}}P_{\varphi _{1k}}+{{%
\textstyle{\frac 12}}}P_{\varphi _{2k}}+{{\textstyle{\frac 12}}}\left[
A_{\varphi _{1k}\varphi _{2k}}+A_{\varphi _{2k}\varphi _{1k}}\right] .
\end{equation*}
Using this to evaluate the left hand side of (\ref{eq5-prod4}), and noting
that $\langle {P}_{\varphi _{1k}}\,,\,{\Pi _{\varphi _1+i\varphi _2}}\rangle
=\langle {P}_{\varphi _{2k}}\,,\,{\Pi _{\varphi _1+i\varphi _2}}\rangle ={%
\textstyle{\frac 12}}$, we can conclude that 
\begin{equation}
\text{Re}\left\langle \,\varphi _{2k}\,|\,\Pi _{\varphi _1+i\varphi
_2}\varphi _{1k}\,\right\rangle =\text{Re}\left\langle \,\varphi
_{2k}\,|\,\varphi _{2k}^{\prime }\,\right\rangle =0.  \label{eq5-orth}
\end{equation}
Next we observe that the operator 
\begin{equation*}
U:=4\left( \Pi _{\varphi _2}\Pi _{\varphi _1+i\varphi _2}\Pi _{\varphi
_1}\right) ^{*}\Pi _{\varphi _2}\Pi _{\varphi _1+\varphi _2}\Pi _{\varphi
_1}=4\Pi _{\varphi _1}\Pi _{\varphi _1+i\varphi _2}\Pi _{\varphi _2}\Pi
_{\varphi _1+\varphi _2}\Pi _{\varphi _1}
\end{equation*}
is a partial isometry that acts as a unitary operator on $\Pi _{\left.
\varphi \right. _1}{\mathcal{H}}$; i.e.\ $U^{*}U=UU^{*}=\Pi _{\left. \varphi
\right. }$. Take the \textsc{onb} $\{\varphi _{1k}\,:\,k=1,\cdots ,m\}$ to
be a set of eigenvectors of this operator, $U\varphi _{1k}=u_k\varphi _{1k}$%
, $|u_k|=1$. It follows that 
\begin{equation*}
\left\langle \,\varphi _{2k}^{\prime }\,|\,\varphi _{2\ell }\,\right\rangle
=\left\langle \,\varphi _{1k}\,|\,U\varphi _{1\ell }\,\right\rangle
=u_k\delta _{k\ell }.
\end{equation*}
Combining this with (\ref{eq5-orth}), it follows that $u_k\in \{+i,-i\}$ for
all $k$. Denote the spectral projections of the partial isometry $U$
associated with the eigenvalues $i,-i$ as $\Pi _{\varphi _1}^{(L)}$, $\Pi
_{\varphi _1}^{(A)}$, respectively. Then the spectral decomposition of $U$
is 
\begin{equation}
U=i\left( \Pi _{\varphi _1}^{(L)}-\Pi _{\varphi _1}^{(A)}\right) ,\ \ \Pi
_{\varphi _1}^{(L)}+\Pi _{\varphi _1}^{(A)}=\Pi _{\varphi _1}.
\label{eq5-pis}
\end{equation}
Finally we can write the operator on the left hand side of (\ref{eq5-prod2})
as 
\begin{equation*}
\Pi _{\varphi _1}\,\Pi _{\varphi _1+i\varphi _2}\Pi _{\varphi _1+\varphi
_2}\Pi _{\varphi _1}=\Pi _{\varphi _1}\,\Pi _{\varphi _1+i\varphi _2}\left(
\Pi _{\varphi _1}+\Pi _{\varphi _2}\right) \Pi _{\varphi _1+\varphi _2}\Pi
_{\varphi _1}={{\textstyle{\frac 14}}}(\Pi _{\varphi _1}+U),
\end{equation*}
where by virtue of (\ref{eq5-pis}) the last expression equals the right hand
side of (\ref{eq5-prod2}).%
\end{proof}

\begin{Lem}
\label{le5-mix3}Let $\dim {\mathcal{H}}>2$ and $T:V\rightarrow \tilde{V}$ be
an $m$-mixing stochastic isometry, $\varphi _1,\varphi _2,\varphi _3\in 
{\mathcal{H}}$ a triple of mutually orthogonal unit vectors, $\alpha ,\beta
\in \Bbb{C}$, $|\alpha |=|\beta |=1$, $k,\ell \in \{1,2,3\}$. Then: 
\begin{eqnarray}
{\Pi _{\varphi _1}\,\Pi _{\varphi _1+i\varphi _2}\Pi _{\varphi _1+\varphi
_2}\Pi _{\varphi _1}\ }={\ {{\textstyle{\frac{1+i}4}}}\,\Pi _{\varphi
_1}^{(L)}\,+\,{{\textstyle{\frac{1-i}4}}}\,\Pi _{\varphi _1}^{(A)};}
\label{eq5-prod5a} \\
{\Pi _{\varphi _1}\,\Pi _{\varphi _1+\alpha \varphi _3}\Pi _{\varphi
_1+\beta \varphi _2}\Pi _{\varphi _1}\ }={\ {{\textstyle{\frac 14}}}\Pi
_{\varphi _1};}  \label{eq5-prod5b} \\
{\Pi _{\varphi _1}\,\Pi _{\varphi _2+\alpha \varphi _3}\Pi _{\varphi
_k+\beta \varphi _\ell }\Pi _{\varphi _1}\ }={\ 0;}  \label{eq5-prod5c} \\
{\Pi _{\varphi _1}\,\Pi _{\varphi _1+i\varphi _3}\Pi _{\varphi _1+\varphi
_3}\Pi _{\varphi _1}\ }={\ {{\textstyle{\frac{1+i}4}}}\,\widetilde{\Pi }%
_{\varphi _1}^{(L)}\,+\,{{\textstyle{\frac{1-i}4}}}\,\widetilde{\Pi }%
_{\varphi _1}^{(A)}.}  \label{eq5-prod5d}
\end{eqnarray}
\end{Lem}

\begin{proof}
We do not claim here that $\widetilde{\Pi }_{\varphi _1}^{(L,A)}=\Pi
_{\varphi _1}^{(L,A)}$, nor that these projections are independent of the
choice of $\varphi _2,\varphi _3$. This will be established in a later step.

Relations (\ref{eq5-prod5a}) and (\ref{eq5-prod5d}) are instances of Lemma 
\ref{le5-mix2}. Due to $\Pi _{\varphi _2+\alpha \varphi _3}\le \Pi _{\varphi
_2}+\Pi _{\varphi _3}$ [cf.\ Eq.\ (\ref{eq5-prod3})], we have $\Pi _{\varphi
_1}\,\Pi _{\varphi _2+\alpha \varphi _3}=0$, which proves (\ref{eq5-prod5c}%
). To verify (\ref{eq5-prod5b}), we compute: 
\begin{eqnarray*}
\Pi _{\varphi _1}\,\Pi _{\varphi _1+\alpha \varphi _3}\Pi _{\varphi _1+\beta
\varphi _2}\Pi _{\varphi _1} &=&\Pi _{\varphi _1}\,\Pi _{\varphi _1+\alpha
\varphi _3}\left( \Pi _{\varphi _1}+\Pi _{\varphi _2}+\Pi _{\varphi
_3}\right) \Pi _{\varphi _1+\beta \varphi _2}\Pi _{\varphi _1} \\
&=&\Pi _{\varphi _1}\,\Pi _{\varphi _1+\alpha \varphi _3}\Pi _{\varphi
_1}\Pi _{\varphi _1+\beta \varphi _2}\Pi _{\varphi _1}={{\textstyle{\frac 14}%
}}\Pi _{\varphi _1}.
\end{eqnarray*}
Here we have used the orthogonalities: $\Pi _{\varphi _1+\alpha \varphi
_3}\Pi _{\varphi _2}=\Pi _{\varphi _3}\Pi _{\varphi _1+\beta \varphi _2}=0$.%
\end{proof}

Approaching the proof of Proposition \ref{pro5-mix2}, we observe that any
triple of pairwise independent unit vectors $\left. \varphi \right. ,\psi
,\vartheta $ spans a subspace $[\varphi ,\psi ,\vartheta ]$ of ${\mathcal{H}}$
of dimension 2 or 3. We assume $\dim {\mathcal{H}}\ge 3$. As will be made
evident, the case $\dim {\mathcal{H}}=2$ can be dealt with by restriction of
the constructions to be carried out for the case $\dim {\mathcal{H}}>2$. Let $%
\varphi _1,\varphi _2,\varphi _3$ be three mutually orthogonal unit vectors
such that $[\varphi ,\psi ,\vartheta ]\subseteq [\varphi _1,\varphi
_2,\varphi _3]=:\mathcal{K}$. Then the nine operators $A_{\varphi _k\varphi
_\ell }$ [cf. (\ref{eq5-pol2})] form a basis of the space of operators on $%
\mathcal{K}$. By polarization, an alternative basis is given by the
operators 
\begin{eqnarray}
&&{P_{\varphi _1},P_{\varphi _2},P_{\varphi _1+\varphi _2},P_{\varphi
_1+i\varphi _2},}  \notag \\
&&{P_{\varphi _3},P_{\varphi _1+\varphi _3},P_{\varphi _2+\varphi
_3},P_{\varphi _1+i\varphi _3},P_{\varphi _2+i\varphi _3}.}  \label{eq5-bas1}
\end{eqnarray}
We denote these $P_{\xi _k}$, $k=1,\cdots ,9$, in the ordering given. It
follows that the corresponding set of projections $\Pi _{\xi _k},\
k=1,\cdots ,9$, i.e.\ 
\begin{eqnarray}
&&{\Pi _{\varphi _1},\Pi _{\varphi _2},\Pi _{\varphi _1+\varphi _2},\Pi
_{\varphi _1+i\varphi _2},}  \notag \\
&&\Pi {_{\varphi _3},\Pi _{\varphi _1+\varphi _3},\Pi _{\varphi _2+\varphi
_3},\Pi _{\varphi _1+i\varphi _3},\Pi _{\varphi _2+i\varphi _3}}
\label{eq5-bas2}
\end{eqnarray}
forms a basis of the space of operators spanned by all $\Pi _{\left. \varphi
\right. }$, $\varphi \in {\mathcal{K}}\backslash \{0\}$.

\noindent \textbf{Proof of Proposition \ref{pro5-mix2}. Step 1.} Excluding
the trivial cases mentioned immediately after Proposition \ref{pro5-mix2},
we need to verify Eq.\ (\ref{eq5-prod1}) for any triple of unit vectors $%
\varphi ,\psi ,\vartheta \in {\mathcal{H}}$ which are pairwise independent and
nonorthogonal. Consider the case $\dim {\mathcal{H}}\ge 3$. We work with a
specific choice of \textsc{onb} $\{\varphi _1,\varphi _2,\varphi _3\}$ of $%
\mathcal{K}$, namely, $\varphi _1=\varphi $, $\varphi _2\perp \varphi _1$
such that $\psi \in [\varphi _1,\varphi _2]$, and $\varphi _3\perp [\varphi
_1,\varphi _2]$ such that $\vartheta \in [\varphi _1,\varphi _2,\varphi _3]=%
\mathcal{K}$. Using the operator bases (\ref{eq5-bas1}) and (\ref{eq5-bas2}%
), and denoting these operators as $P_{\xi _k},\Pi _{\xi _k}$ in the order
given, we can write 
\begin{equation*}
P{_{\left. \psi \right. }}=\sum_{k=1}^4s_kP_{\xi _k},\ \ P_{\left. \vartheta
\right. }=\sum_{\ell =1}^9t_\ell P_{\xi _\ell },
\end{equation*}
and 
\begin{equation*}
\Pi {_{\left. \psi \right. }}=\sum_{k=1}^4s_k\Pi _{\xi _k},\ \ \Pi _{\left.
\vartheta \right. }=\sum_{\ell =1}^9t_\ell \Pi _{\xi _\ell },
\end{equation*}
with all $s_k,t_\ell \in \Bbb{R}$. With this we find: 
\begin{eqnarray}
{P_{\left. \varphi \right. }}\ {P_{\left. \vartheta \right. }}\ {P_{\left.
\psi \right. }P_{\left. \varphi \right. }\ }\ \ &=&{\ \left\langle \,\varphi
\,|\,\vartheta \,\right\rangle \,\left\langle \,\vartheta \,|\,\psi
\,\right\rangle \,\left\langle \,\psi \,|\,\varphi \,\right\rangle
\,P_{\left. \varphi \right. }}  \notag \\
&=&{\sum_{k=1}^4\sum_{\ell =1}^9s_kt_\ell P_{\left. \varphi \right. }P_{\xi
_\ell }P_{\xi _k}P_{\left. \varphi \right. }}  \label{eq5-brace} \\
&=&{\{s_1t_1+{{\textstyle{\frac 12}}}s_1t_3+{{\textstyle{\frac 12}}}s_1t_4+{{%
\textstyle{\frac 12}}}s_1t_6+{{\textstyle{\frac 12}}}s_1t_8}  \notag \\
&&{+{{\textstyle{\frac 12}}}s_3t_1+{{\textstyle{\frac 12}}}s_3t_3+{{%
\textstyle{\frac{1+i}4}}}s_3t_4}+\ {{{\textstyle{\frac 14}}}s_3t_6+{{%
\textstyle{\frac 14}}}s_3t_8}  \notag \\
&&{+{{\textstyle{\frac 12}}}s_4t_1+{{\textstyle{\frac{1-i}4}}}s_4t_3+{{%
\textstyle{\frac 12}}}s_4t_4+{{\textstyle{\frac 14}}}s_4t_6+{{\textstyle{%
\frac 14}}}s_4t_8\}\,P_{\left. \varphi \right. }.}  \notag
\end{eqnarray}
Next, applying Eqs.\ (\ref{eq5-prod5a},\ref{eq5-prod5b},\ref{eq5-prod5c})
[the expression (\ref{eq5-prod5d}) does not occur], we evaluate the
corresponding sum for the left hand side of Eq.\ (\ref{eq5-prod1}): 
\begin{equation}
\Pi _{\left. \varphi \right. }\Pi {_{\left. \vartheta \right. }}\Pi {%
_{\left. \psi \right. }}\Pi _{\left. \varphi \right. }\ =\
\sum_{k=1}^4\sum_{\ell =1}^9s_kt_\ell \Pi _{\left. \varphi \right. }\Pi
_{\xi _\ell }\Pi _{\xi _k}\Pi _{\left. \varphi \right. }=\ a\Pi _{\varphi
_1}^{(L)}+\overline{a}\Pi _{\varphi _1}^{(A)},  \label{eq5-spec1}
\end{equation}
where $a$ turns out to be the same expression as that given in the braces $%
\{\cdots \}$ in Eq.\ (\ref{eq5-brace}). Hence, $a=\left\langle \,\varphi
\,|\,\vartheta \,\right\rangle \,\left\langle \,\vartheta \,|\,\psi
\,\right\rangle \,\left\langle \,\psi \,|\,\varphi \,\right\rangle $, and
Eq.\ (\ref{eq5-prod1}) is verified. Note that the case $\dim {\mathcal{H}}=2$
is covered by putting $s_k=t_k=0$ for $k>4$.

\textbf{Step 2}{.} Next we show that the spectral projections occurring in
Eq.\ (\ref{eq5-prod1}) do not depend on the choice of $\psi ,\vartheta \in 
{\mathcal{H}}$. Note that the result of Step 1 holds for any choice of unit
vectors $\vartheta \in {\mathcal{H}}$, and thus for any choice of $\varphi
_3\perp [\varphi _1,\varphi _2]$ in the case $\dim {\mathcal{H}}\ge 3$. The
construction of $\Pi _{\varphi _1}^{(L)},\Pi _{\varphi _1}^{(A)}$ so far
depends on the choice of $\psi \in {\mathcal{H}}$, but only via $\varphi
_2\perp \varphi _1$. In the case $\dim {\mathcal{H}}=2$, the ray orthogonal to 
$\varphi _1$ is uniquely determined by that vector, so that the spectral
projections in Eq. (\ref{eq5-spec1}) are the same for all $\psi ,\vartheta $%
. In the case $\dim {\mathcal{H}}\ge 3$, we have to show that these
projections are actually the same for \textsl{all} choices of unit vectors $%
\varphi _2^{\prime }\perp \varphi _1$ [including $\varphi _2^{\prime
}=\varphi _3$, cf.\ Eq.\ (\ref{eq5-prod5d})]. To this end we repeat the
procedure of Step 1, this time choosing a unit vector $\varphi _2^{\prime
}\perp \varphi _1=\varphi $ in the plane $[\varphi ,\vartheta ]$. Then there
exists a unit vector $\varphi _3^{\prime }$ such that $\psi \in [\varphi
_1,\varphi _2^{\prime },\varphi _3^{\prime }]=\mathcal{K}$. We obtain again
a spectral decomposition of the form (\ref{eq5-spec1}), 
\begin{equation}
\Pi _{\left. \varphi \right. }\Pi {_{\left. \vartheta \right. }}\Pi {%
_{\left. \psi \right. }}\Pi _{\left. \varphi \right. }=a{\Pi _{\left.
\varphi \right. }^{\prime }}^{(L)}\,+\,\overline{a}{\Pi _{\left. \varphi
\right. }^{\prime }}^{(A)},\ \ {\Pi _{\left. \varphi \right. }^{\prime }}%
^{(L)}\,+\,{\Pi _{\left. \varphi \right. }^{\prime }}^{(A)}=\Pi _{\left.
\varphi \right. },  \label{eq5-spec2}
\end{equation}
for the \textsl{same} set of vectors $\varphi ,\psi ,\vartheta $. Provided
that the eigenvalues $a,\overline{a}$ do not coincide, i.e.\ $a\notin \Bbb{R}
$, then the spectral projections ${\Pi _{\left. \varphi \right. }^{\prime }}%
^{(L)},{\Pi _{\left. \varphi \right. }^{\prime }}^{(A)}$ constructed from $%
\varphi _1,\varphi _2^{\prime }$ along the lines of the proof of Lemma \ref
{le5-mix2} coincide with the projections $\Pi _{\left. \varphi \right.
}^{(L)},\Pi _{\left. \varphi \right. }^{(A)}$ constructed from $\varphi
_1,\varphi _2$.

We show that, given $\varphi ,\psi $, any $\varphi _2^{\prime }\perp \varphi
_1$ can be realized with a choice of $\vartheta $ such that $a=\left\langle
\,\varphi _1\,|\,\vartheta \,\right\rangle \,\left\langle \,\vartheta
\,|\,\psi \,\right\rangle \,\left\langle \,\psi \,|\,\varphi
_1\,\right\rangle \notin \Bbb{R}$. Let $\varphi _2^{\prime }$ be any unit
vector such that $\varphi _2^{\prime }\perp \varphi $ and $\varphi
_2^{\prime }\not{\perp}\psi $. Choose $\vartheta =\alpha _1\varphi _1+\alpha
_2\varphi _2^{\prime }$, with $\alpha _1,\alpha _2$ to be further specified
shortly. Take a unit vector $\varphi _3^{\prime }$ orthogonal to $\varphi
_1,\varphi _2^{\prime }$ and such that $\psi =\beta _1\varphi _1+\beta
_2\varphi _2^{\prime }+\beta _3\varphi _3^{\prime }$. Since $\varphi
_2^{\prime }\not{\perp}\psi $, we have $\beta _1\ne 0\ne \beta _2$. We
obtain 
\begin{equation*}
a\ =\ \alpha _1\,\{\overline{\alpha }_1\beta _1\,+\,\overline{\alpha }%
_2\beta _2\}\,\overline{\beta }_1\ =\ |\alpha _1|^2\,|\beta _1|^2\,+\,\alpha
_1\overline{\alpha }_2\beta _2\beta _1.
\end{equation*}
The first term on the right hand side is real and positive. Then $a\ne \Bbb{R%
}$ can be easily achieved by proper choice of $\alpha _1,\alpha _2$. Hence
all choices of $\varphi _2^{\prime }\perp \varphi _1$, $\varphi _2^{\prime
}\ne \varphi _3$ lead to the same spectral decomposition in Eqs.\ (\ref
{eq5-prod2}) or (\ref{eq5-prod5a}). [Note that $\varphi _2^{\prime }=\varphi
_3$ would imply $\beta _2=0$.]

A continuity argument finally shows that the case $\varphi _2^{\prime
}=\varphi _3$ can be included, too, so that the spectral decompositions in
Eqs.\ (\ref{eq5-prod5a}) and (\ref{eq5-prod5d}) do indeed coincide. In fact
let $\varphi _2^{(n)}$ be a sequence of unit vectors orthogonal to $\varphi
_1$ and different from $\varphi _3$, such that $||${$\varphi $}${_2^{(n)}-}${%
$\varphi $}${_3}||\to 0$. Then $P_{\varphi _1+\alpha \varphi _2^{(n)}}$
converges to $P_{\varphi _1+\alpha \varphi _3}$ in trace norm. By the
continuity of $T$, $\Pi _{\varphi _1+\alpha \varphi _2^{(n)}}$ converges to $%
\Pi _{\varphi _1+\alpha \varphi _3}$ in trace norm. Then 
\begin{equation*}
{{\textstyle{\frac{1+i}4}}}\,\Pi _{\varphi _1}^{(L)}\,+\,{{\textstyle{\frac{%
1-i}4}}}\,\Pi _{\varphi _1}^{(A)}\ =\ \Pi _{\varphi _1}\,\Pi _{\varphi
_1+i\varphi _2^{(n)}}\Pi _{\varphi _1+\varphi _2^{(n)}}\Pi _{\varphi _1}\
\longrightarrow \ \Pi _{\varphi _1}\,\Pi _{\varphi _1+i\varphi _3}\Pi
_{\varphi _1+\varphi _3}\Pi _{\varphi _1},
\end{equation*}
convergence in trace norm. It follows that the first and last expressions
coincide. This concludes Step 2.%
\endproof%

We introduce the following notation for the operators of Eq.\ (\ref
{eq5-prod1}): 
\begin{eqnarray*}
{\mathcal{W}_{\vartheta \psi }^{\left. \varphi \right. }\ } &=&{\ \Pi
_{\left. \varphi \right. }\Pi _{\left. \vartheta \right. }\Pi _{\left. \psi
\right. }\Pi _{\left. \varphi \right. }\ =a_{\vartheta \psi }^{\left.
\varphi \right. }\,\Pi _{\left. \varphi \right. }^{(L)}\,+\,\overline{%
a_{\vartheta \psi }^{\left. \varphi \right. }}\,\Pi _{\left. \varphi \right.
}^{(A)},} \\
{a_{\vartheta \psi }^{\left. \varphi \right. }\ } &=&{\ \left\langle
\,\varphi \,|\,\vartheta \,\right\rangle \,\left\langle \,\vartheta
\,|\,\psi \,\right\rangle \,\left\langle \,\psi \,|\,\varphi \right\rangle }.
\end{eqnarray*}

\begin{Prop}
\label{pro5-mix3}Let $T:V\rightarrow \tilde{V}$ be an $m$-mixing stochastic
isometry. For any two unit vectors $\left. \varphi \right. ,\chi \in 
{\mathcal{H}}$ one has 
\begin{eqnarray}
\quad {\Pi _{\left. \chi \right. }\Pi _{\left. \varphi \right. }^{(L)}\Pi
_{\left. \chi \right. }\ }={\ \langle {P_{\left. \chi \right. }}\,,\,{%
P_{\left. \varphi \right. }}\rangle \Pi _{\left. \chi \right. }^{(L)},\ \
\Pi _{\left. \chi \right. }\Pi _{\left. \varphi \right. }^{(A)}\Pi _{\left.
\chi \right. }\ =\ \langle {P_{\left. \chi \right. }}\,,\,{P_{\left. \varphi
\right. }}\rangle \Pi _{\left. \chi \right. }^{(A)};}  \label{eq5-la1} \\
{\Pi _{\left. \varphi \right. }^{(L)}\Pi _{\left. \chi \right. }^{(A)}\ }={\
\Pi _{\left. \varphi \right. }^{(A)}\Pi _{\left. \chi \right. }^{(L)}\ =\ 0.}
\label{eq5-la2}
\end{eqnarray}
The projections 
\begin{equation}
P^{(L)}=\bigvee_{\varphi \in {\mathcal{H}}\backslash \{0\}}\Pi _{\left.
\varphi \right. }^{(L)},\ \ \ \ P^{(A)}=\bigvee_{\varphi \in {\mathcal{H}}%
\backslash \{0\}}\Pi _{\left. \varphi \right. }^{(A)}  \label{eq5-la3}
\end{equation}
are mutually orthogonal. Moreover, for all $\left. \varphi \right. \in 
{\mathcal{H}}\backslash \{0\}$, 
\begin{equation}
\Pi _{\left. \varphi \right. }^{(L)}\ =\ P^{(L)}\Pi _{\left. \varphi \right.
},\ \ \ \ \Pi _{\left. \varphi \right. }^{(A)}\ =\ P^{(A)}\Pi _{\left.
\varphi \right. },  \label{eq5-la4}
\end{equation}
and the ranks $m^{(L)},m^{(A)}$ of the projections $\Pi _{\left. \varphi
\right. }^{(L)}$, $\Pi _{\left. \varphi \right. }^{(A)}$ are independent of $%
\varphi \in {\mathcal{H}}\backslash \{0\}$ and satisfy $m^{(L)}+m^{(A)}=m$. 
\newline
If the invariants $m^{(L)},m^{(A)}$ are nonzero, then $T$ decomposes into a
convex combination of two stochastic isometries $T^{(L)},T^{(A)}$ with
orthogonal ranges, 
\begin{equation}
T\ =\ {{\textstyle{\frac{m^{(L)}}m}}}\,T^{(L)}\,+\,{{\textstyle{\frac{m^{(A)}%
}m}}}\,T^{(A)},  \label{eq5-la5}
\end{equation}
\begin{equation}
T^{(L)}(\rho )\ =\ {{\textstyle{\frac 1{m^{(L)}}}}}\,P^{(L)}\,T(\rho
)\,P^{(L)},\ \ \ T^{(A)}(\rho )\ =\ {{\textstyle{\frac 1{m^{(A)}}}}}%
\,P^{(A)}\,T(\rho )\,P^{(A)}.  \label{eq5-la6}
\end{equation}
\end{Prop}

\begin{proof}
Let $\chi \in {\mathcal{H}}$ be a unit vector not orthogonal to $\varphi $,
and choose any unit vector $\psi \in {\mathcal{H}}$ neither orthogonal nor
parallel to $\varphi ,\chi $ and such that $a_{\chi \psi }^{\left. \varphi
\right. }\notin \Bbb{R}$. We compute the operator $\Pi _{\left. \chi \right.
}\mathcal{W}_{\chi \psi }^{\left. \varphi \right. }\Pi _{\left. \chi \right.
}$ in two ways, using (\ref{eq5-mix3}) and noting that $a_{\chi \psi
}^{\left. \varphi \right. }=a_{\psi \left. \varphi \right. }^\chi $: 
\begin{eqnarray*}
{\Pi _{\left. \chi \right. }\mathcal{W}_{\chi \psi }^{\left. \varphi \right.
}\Pi _{\left. \chi \right. }\ =\ \langle {P_{\left. \chi \right. }}\,,\,{%
P_{\left. \varphi \right. }}\rangle \mathcal{W}_{\psi \left. \varphi \right.
}^\chi \ } &=&{\ \langle {P_{\left. \chi \right. }}\,,\,{P_{\left. \varphi
\right. }}\rangle \left( a_{\chi \psi }^{\left. \varphi \right. }\,\Pi
_{\left. \chi \right. }^{(L)}\,+\,\overline{a_{\chi \psi }^{\left. \varphi
\right. }}\,\Pi _{\left. \chi \right. }^{(A)}\right) } \\
&=&a_{\chi \psi }^{\left. \varphi \right. }\Pi _{\left. \chi \right. }\Pi
_{\left. \varphi \right. }^{(L)}\Pi _{\left. \chi \right. }+\overline{%
a_{\chi \psi }^{\left. \varphi \right. }}\Pi _{\left. \chi \right. }\Pi
_{\left. \varphi \right. }^{(A)}\Pi _{\left. \chi \right. }.
\end{eqnarray*}
Using the fact that the map $\Pi _{\left. \chi \right. }|_{\Pi _{\left.
\varphi \right. }{\mathcal{H}}}:\Pi _{\left. \varphi \right. }{\mathcal{H}}\to
\Pi _{\left. \chi \right. }{\mathcal{H}}$ is orthogonality preserving [cf.\
Eq.\ (\ref{eq5-mix4})], Eqs.\ (\ref{eq5-la1}) and (\ref{eq5-la2}) follow by
application of Eq.\ (\ref{eq5-mix2}) and a version of (\ref{eq5-mix4}). The
orthogonality of $P^{(L)},P^{(A)}$ and the relation (\ref{eq5-la4}) are then
obvious.

The invariance of the numbers $m^{(L)},m^{(A)}$ is a consequence of the fact
that $\Pi _{\left. \chi \right. }$ transforms an orthogonal basis of $\Pi
_{\left. \varphi \right. }^{(L)}{\mathcal{H}}$ [$\Pi _{\left. \varphi \right.
}^{(A)}{\mathcal{H}}$] onto an orthogonal basis of $\Pi _{\left. \chi \right.
}^{(L)}{\mathcal{H}}$ [$\Pi _{\left. \chi \right. }^{(A)}{\mathcal{H}}$].

For $\varphi \in {\mathcal{H}}\backslash \{0\}$ we have, by virtue of Eq.\ (%
\ref{eq5-la4}): 
\begin{eqnarray*}
T(P_{\left. \varphi \right. })\ &=&\ {{\textstyle{\frac 1m}}}\,\Pi _{\left.
\varphi \right. }\ =\ {{\textstyle{\frac 1m}}}\left( P^{(L)}\,\Pi _{\left.
\varphi \right. }\,P^{(L)}\,+\,P^{(A)}\,\Pi _{\left. \varphi \right.
}\,P^{(A)}\right) \\
&=&\ {{\textstyle{\frac{m^{(L)}}m}}}\,T^{(L)}(P_{\left. \varphi \right.
})\,+\,{{\textstyle{\frac{m^{(A)}}m}}}\,T^{(A)}(P_{\left. \varphi \right. }).
\end{eqnarray*}
Due to the linearity and continuity of $T$ this equation extends to all ($%
\sigma $-)\allowbreak {convex} combinations of pure states and thus to all
states $\rho $ and finally to all of $V$.%
\end{proof}

Next we come to analyze the stochastic isometries $T^{(L)}$, $T^{(A)}$, or
equivalently, the $m$-mixing stochastic isometries for which either $%
m=m^{(L)}$ or $m=m^{(A)}$. We note that 
\begin{eqnarray}
{\mathcal{W}_{\vartheta \psi }^{\left. \varphi \right. }\ } &=&{\ \Pi
_{\left. \varphi \right. }\Pi _{\left. \vartheta \right. }\Pi _{\left. \psi
\right. }\Pi _{\left. \varphi \right. }\ =\ a_{\vartheta \psi }^{\left.
\varphi \right. }\,\Pi _{\left. \varphi \right. },\ \ \ [\text{case}\
m=m^{(L)}]}  \label{eq5-w1} \\
{\mathcal{W}_{\vartheta \psi }^{\left. \varphi \right. }\ } &=&{\ \Pi
_{\left. \varphi \right. }\Pi _{\left. \vartheta \right. }\Pi _{\left. \psi
\right. }\Pi _{\left. \varphi \right. }\ =\ \overline{a_{\vartheta \psi
}^{\left. \varphi \right. }}\,\Pi _{\left. \varphi \right. }.\ \ \ [\text{%
case}\ m=m^{(A)}]}  \label{eq5-w2}
\end{eqnarray}

Let $T$ be an $m$-mixing stochastic isometry with $m=m^{(L)}$ or $m=m^{(A)}$%
. Pick an arbitrary vector $\varphi \in {\mathcal{H}}\backslash \{0\}$ and an
orthonormal basis $\{\varphi _1,\cdots ,\varphi _m\}$ of $\Pi _{\left.
\varphi \right. }{\mathcal{H}}$. For any $\psi \in {\mathcal{H}}\backslash \{0\}$
with $\psi \not{\perp}\varphi $, the set of vectors 
\begin{equation}
\psi _k:=\langle {P_{\left. \psi \right. }}\,,\,{P_{\left. \varphi \right. }}%
\rangle ^{-1/2}\Pi {_{\left. \psi \right. }}\varphi _k  \label{eq5-obas}
\end{equation}
is an orthonormal basis of $\Pi {_{\left. \psi \right. }}{\mathcal{H}}$, cf.\
Eq.\ (\ref{eq5-mix4}). Moreover, for any two unit vectors $\psi ,\vartheta $
not orthogonal to $\varphi $, the corresponding basis vectors $\psi _k$, $%
\vartheta _\ell $ are mutually orthogonal if $k\ne \ell $. In fact, due to
Eqs.\ (\ref{eq5-w1},\ref{eq5-w2}) we have ($\sim $ denoting proportionality) 
\begin{equation}
\left\langle \,\vartheta _\ell \,|\,\psi _k\,\right\rangle \sim \left\langle
\,\Pi _{\left. \vartheta \right. }\varphi _\ell \,|\,\Pi {_{\left. \psi
\right. }}\varphi _k\,\right\rangle =\left\langle \,\varphi _k\,|\,\mathcal{W%
}_{\vartheta \psi }^{\left. \varphi \right. }\,\varphi _\ell \,\right\rangle
\sim \delta _{k\ell }.  \label{eq5-w3}
\end{equation}

\begin{Prop}
\label{pro5-mix4}Let $T$ be an $m$-mixing stochastic isometry with $m=m^{(L)}
$ or $m=m^{(A)}$. Fix a vector $\varphi ^0\in {\mathcal{H}}%
\backslash \{0\}$ and an orthonormal basis $\{\varphi _1,\cdots ,\varphi _m\}
$ of $\Pi _{\varphi ^0}{\mathcal{H}}$. The projections 
\begin{equation}
P_k:=\bigvee_{\psi \not{\perp}\varphi ^0}P_{\psi _k}.  \label{eq5-pk1}
\end{equation}
are mutually orthogonal. For any $\psi \in {\mathcal{H}}\backslash \{0\}$
there exists a unique ray $P_{\psi _k}{\mathcal{H}}$, $\psi _k\in \Pi _{\left.
\psi \right. }{\mathcal{H}}$, such that 
\begin{equation}
P_k\Pi _{\left. \psi \right. }\ =\ \Pi _{\left. \psi \right. }P_k\ =\
P_{\psi _k},\ \ k=1,\cdots ,m.  \label{eq5-pk2}
\end{equation}
Then 
\begin{equation}
P_k\ =\ \bigvee_{\psi \in {\mathcal{H}}\backslash \{0\}}P_{\psi _k},\ \
k=1,\cdots ,m,  \label{eq5-pk3}
\end{equation}
and 
\begin{equation}
\sum_{k=1}^mP_k\ =\ \bigvee_{\left. \varphi \right. \in {\mathcal{H}}%
\backslash \{0\}}\Pi _{\left. \varphi \right. }.  \label{eq5-pk4}
\end{equation}
$T$ decomposes into a convex combination of pure stochastic isometries, 
\begin{equation}
T\ =\ \sum_{k=1}^m{{\textstyle{\frac 1m}}}T_k,  \label{eq5-tk1}
\end{equation}
\begin{equation}
T_k(\rho )\ =\ m\,P_k\,T(\rho )\,P_k.  \label{eq5-tk2}
\end{equation}
Hence there exist maps $U_k:{\mathcal{H}}\to P_k\tilde{{\mathcal{H}}}$, $%
k=1,\cdots ,m$, which are all unitary if $m=m^{(L)}$ or all antiunitary if $%
m=m^{(A)}$, such that $P_k=U_kU_k^{*}$ and 
\begin{equation}
T_k(\rho )\ =\ U_k\,\rho \,U_k^{*},\ \ \rho \in {\mathcal{H}}_1({\mathcal{H}})_s.
\label{eq5-tk3}
\end{equation}
\end{Prop}

\begin{proof}
The orthogonality of the $P_k$ is a direct consequence of Eq.\ (\ref{eq5-w3}%
). Let $\psi $ be a nonzero vector. If $\psi \not{\perp}\varphi ^0$, then
take $\psi _k\sim \Pi _{\left. \psi \right. }\varphi _k$. We have $P_{\psi
_k}\le P_k$ for all $k$, thus $P_{\psi _k}P_k=P_{\psi _k}$ and, due to the
orthogonality of the $P_k$, $P_{\psi _\ell }P_k=0$ if $\ell \ne k$. We
obtain $P_k\Pi _{\left. \psi \right. }=P_k\sum_\ell P_{\psi _\ell
}=P_kP_{\psi _k}=P_{\psi _k}$. This proves (\ref{eq5-pk2}) for the case $%
\psi \not{\perp}\varphi ^0$.

Consider the case $\psi \perp \varphi ^0$. We assume that $\varphi ^0,\psi $
are unit vectors. Define a sequence of unit vectors $\vartheta
^{(n)}:=\left( \frac 1n\right) ^{1/2}\varphi ^0+\left( 1-\frac 1n\right)
^{1/2}\psi \to \psi $ as $n\to \infty $. We can define an \textsc{onb} of $%
\Pi _{\vartheta ^{(n)}}{\mathcal{H}}$ as in (\ref{eq5-obas}), $\vartheta
_k^{(n)}\sim \Pi _{\vartheta ^{(n)}}\varphi _k$. Using again Eq.\ (\ref
{eq5-w1}),(\ref{eq5-w2}), it is easy to verify that for each $k$, the
vectors $\vartheta _k^{(n)}$ form a Cauchy sequence. Let $\psi _k$ denote
the limiting unit vector. We show that $P_k\Pi _{\left. \psi \right.
}=P_{\psi _k}$. (This ensures the uniqueness of the ray $P_{\psi _k}$.)

Let $\xi \in {\mathcal{H}}$ be any vector. One estimates 
\begin{eqnarray*}
\left\| \left( P_{\vartheta _k^{(n)}}-P_{\psi _k}\right) \xi \right\| \ &\le
&||{\vartheta _k^{(n)}-\psi _k}\,||\ |\langle {\vartheta _k^{(n)}}\,|\,\xi
\rangle |+\,\left\| \psi _k\right\| \ |\langle {\vartheta _k^{(n)}-\psi _k}%
\,|\,\xi \rangle | \\
&\le &\ 2\,||{\vartheta _k^{(n)}-\psi _k}||\,\left\| \xi \right\| .
\end{eqnarray*}
We conclude that $\big\Vert{P_{\vartheta _k^{(n)}}-P_{\psi _k}}\big\Vert\to
0 $ and so $\big\Vert{P_{\vartheta _k^{(n)}}-P_{\psi _k}}\big\Vert_1\to 0$.
We also have $\left\| P_{\vartheta ^{(n)}}-P_{\left. \psi \right. }\right\|
_1\to 0$ and therefore, due to the continuity of $T$, $\left\| \Pi
_{\vartheta ^{(n)}}-\Pi _{\left. \psi \right. }\right\| _1\to 0$ and also $%
\left\| P_k\Pi _{\vartheta ^{(n)}}-P_k\Pi _{\left. \psi \right. }\right\|
_1\to 0$. But $P_k\Pi _{\vartheta ^{(n)}}=P_{\vartheta _k^{(n)}}\to P_{\psi
_k}$, and so $P_k\Pi _{\left. \psi \right. }=P_{\psi _k}$. This proves (\ref
{eq5-pk2}).

The relation (\ref{eq5-pk3}) is an immediate consequence of the fact just
demonstrated that for all $\psi $, $P_{\psi _k}\le P_k$. Finally we have $%
\left( \sum_kP_k\right) \Pi _{\left. \varphi \right. }=\Pi _{\left. \varphi
\right. }$ and so $\Pi _{\left. \varphi \right. }\le \sum_kP_k$. Since all $%
P_{\psi _k}\le \Pi _{\left. \psi \right. }$, the converse ordering holds as
well. This proves (\ref{eq5-pk4}).

The maps $T_k$ are clearly linear and positive. We have $T_k(P_{\left. \xi
\right. })=P_k\Pi _{\left. \xi \right. }=P_{\xi _k}$ for all vectors $\xi
\ne 0$, so the $T_k$ are trace preserving and thus pure stochastic maps. Let 
$\psi ,\xi \in {\mathcal{H}}\backslash \{0\}$, with $\psi \perp \xi $. Then $%
\Pi _{\left. \psi \right. }\perp \Pi _{\left. \xi \right. }$ [since $T$ is
orthogonality preserving], and so $P_{\psi _k}\perp P_{\xi _k}$, that is, $%
T_k(P_{\left. \psi \right. })\perp T_k(P_{\left. \xi \right. })$. Thus the $%
T_k$ are orthogonality preserving and therefore (pure stochastic)
isometries. Eq.\ (\ref{eq5-tk1}) then follows by straightforward application
of (\ref{eq5-tk2}). The existence of unitary or antiunitary maps $U_k$
satisfying (\ref{eq5-tk3}) is a consequence of Theorem 2.3.1 of \cite{Dav}
[cf. Proposition \ref{pro3-pure} above].

If $U_k$ is unitary, we find, using $U_k^{*}U_k=I$: 
\begin{equation*}
T_k(P_{\left. \varphi \right. })T_k(P_{\left. \vartheta \right.
})T_k(P_{\left. \psi \right. })T_k(P_{\left. \varphi \right. })\ =\
U_k(P_{\left. \varphi \right. }P_{\left. \vartheta \right. }P_{\left. \psi
\right. }P_{\left. \varphi \right. })\,U_k^{*}\ =\ a_{\vartheta \psi
}^{\left. \varphi \right. }\,T_k(P_{\left. \varphi \right. }).
\end{equation*}
Similarly, if $U_k$ is antiunitary, we find: 
\begin{equation*}
T_k(P_{\left. \varphi \right. })T_k(P_{\left. \vartheta \right.
})T_k(P_{\left. \psi \right. })T_k(P_{\left. \varphi \right. })\ =\
U_k(P_{\left. \varphi \right. }P_{\left. \vartheta \right. }P_{\left. \psi
\right. }P_{\left. \varphi \right. })\,U_k^{*}\ =\ \overline{a_{\vartheta
\psi }^{\left. \varphi \right. }}\,T_k(P_{\left. \varphi \right. }).
\end{equation*}
Observing that $T_k(P_{\left. \xi \right. })\perp T_\ell (P_{\left. \chi
\right. })$ ($k\ne \ell $) and $\Pi _{\left. \varphi \right. }=\sum
T_k(P_{\left. \varphi \right. })$, etc., we obtain: 
\begin{equation*}
{\mathcal{W}}_{\vartheta \psi }^{\left. \varphi \right. }\ =\ \Pi _{\left.
\varphi \right. }\Pi _{\left. \vartheta \right. }\Pi _{\left. \psi \right.
}\Pi _{\left. \varphi \right. }\ =\ \sum_kT_k(P_{\left. \varphi \right.
})T_k(P_{\left. \vartheta \right. })T_k(P_{\left. \psi \right.
})T_k(P_{\left. \varphi \right. })\ =\ \sum_ka_kT_k(P_{\left. \varphi
\right. }),
\end{equation*}
where $a_k={a_{\vartheta \psi }^{\left. \varphi \right. }}$ or $a_k=%
\overline{a_{\vartheta \psi }^{\left. \varphi \right. }}$ according to
whether $U_k$ is unitary or antiunitary. Comparing this with Eqs.\ (\ref
{eq5-w1}),(\ref{eq5-w2}), it is seen that all the $U_k$ must either be
simultaneously unitary or antiunitary in order to reproduce the right hand
sides of these equations.

From $T_k(P_{\left. \varphi \right. })\cdot T_\ell (P_{\left. \psi \right.
})=0$ for $k\ne \ell $, $\varphi ,\psi \in {\mathcal{H}}\backslash \{0\}$, it
follows that $U_k^{*}U_\ell =\delta _{k\ell }I_{{\mathcal{H}}}$. Also, it is
obvious that the projection $U_kU_k^{*}$ has the same range as $P_k$, and so 
$U_kU_k^{*}=P_k$.%
\end{proof}

\begin{Rem} The stochastic isometries $T^{(L)}$, $T^{(A)}$
of Proposition \ref{pro5-mix3}, Eq.\ (\ref{eq5-la6}), are convex
combinations of pure stochastic isometries generated by linear and
antilinear isometries of ${\mathcal{H}}$, respectively. This explains the use
of the superscripts $(L),(A)$ throughout this section.
\end{Rem}

\noindent \textbf{Proofs of Proposition \ref{pro3-mix} and Theorem \ref
{th3-is}.} Proposition \ref{pro3-mix} is a direct consequence of the
combination of Propositions \ref{pro5-mix3} and \ref{pro5-mix4}. Theorem \ref
{th3-is} is a direct result of the combination of Propositions \ref{pro3-dec}
and \ref{pro3-mix}.%
{\ \ \ $\square$}

\noindent \textbf{Proof of Theorem \ref{th3-cp}}{.} Let $T=\sum_kw_kT_k$ be
a stochastic isometry of the form of Eq.\ (\ref{eq3-is}), the $T_k$ being
pure stochastic isometries generated by unitary or antiunitary maps $U_k:%
{\mathcal{H}}\to \tilde{{\mathcal{H}}}_k$. For $U_k$ unitary (antiunitary), $T_k$
is completely positive (not completely positive). The sum of completely
positive maps is completely positive, and this extends to $\sigma $-convex
combinations. Hence $T$ is completely positive whenever all $U_k$ are
unitary.

Conversely, let $T$ be completely positive. Suppose some $U_k$ is
antiunitary. Hence there exist unit vectors $\Psi \in \tilde{{\mathcal{H}}}%
_k\otimes {\Bbb{C}}^n$, $\Theta \in {\mathcal{H}}\otimes {\Bbb{C}}^n$ for some $%
n\in {\Bbb{N}}$ such that 
\begin{equation}
\left\langle \,\Psi \,|\,T\otimes \imath (P_{\Theta} )\Psi \,\right\rangle
=\left\langle \,\Psi \,|\,T_k\otimes \imath (P_{\Theta} )\Psi
\,\right\rangle <0.  \label{eq5-cp}
\end{equation}
(Here $\imath $ denotes the identity map on 
${\mathcal{H}}_1({\Bbb{C}}^n)$.) The
equality is due to the choice of $\Psi $ in the subspace 
$\tilde{{\mathcal{H}}}_k\otimes {\Bbb{C}}^n$ and the fact that 
the ranges of the $U_\ell $ are
mutually orthogonal. Eq.\ (\ref{eq5-cp}) contradicts the complete positivity
of $T$.%
{\ \ \ $\square$}

\section{Concluding Remarks. Some Physical Applications}

Stochastic isometries give rise to a variety of associated maps some of
which will be briefly described here. Any linear bounded map 
$T:{\mathcal{B}}%
_1({\mathcal{H}})_s\to {\mathcal{B}}_1(\tilde{{\mathcal{H}}})_s$ has a unique
extension to a linear bounded map on ${\mathcal{B}}_1({\mathcal{H}})$. This
extension of $T$ will be denoted $\widehat{T}$. A stochastic map $T$ has a
unique extension to a linear, trace preserving map on ${\mathcal{B}}_1(%
{\mathcal{H}})$.

The extension $\widehat{T}$ to ${\mathcal{B}}_1({\mathcal{H}})$ of a pure
stochastic isometry $T$ generated by a unitary or antiunitary map $U:%
{\mathcal{H}}\to \tilde{{\mathcal{H}}}$ is given as follows: for $\rho \in 
{\mathcal{B}}_1({\mathcal{H}})$, 
\begin{eqnarray*}
{{\widehat{T}}(\rho )\,} &&{=\,U\rho \,U^{*}\ \ \ \ \text{if}\ U\ \text{is
unitary};} \\
{{\widehat{T}}(\rho )\,} &\,&{=\,U\rho ^{*}U^{*}\ \ \ \text{if}\ U\ \text{is
antiunitary}.}
\end{eqnarray*}

\begin{Prop}
\label{pro6-ext}The extension $\widehat{T}$ :${\mathcal{B}}_1({\mathcal{H}}%
)\rightarrow {\mathcal{B}}_1(\tilde{{\mathcal{H}}})$ of a 
stochastic isometry $T:%
{\mathcal{B}}_1({\mathcal{H}})_s\to 
{\mathcal{B}}_1(\tilde{{\mathcal{H}}})_s$ is a
trace preserving isometry. Conversely, the restriction to ${\mathcal{B}}_1(%
{\mathcal{H}})_s$ of a trace preserving isometry on ${\mathcal{B}}_1({\mathcal{H}})
$ is a stochastic isometry.
\end{Prop}

\begin{proof}
Let $T$ be a stochastic isometry, $\widehat{T}$ its extension. That $%
\widehat{T}$ is trace preserving follows trivially from the corresponding
property or $T$. Using Eq.\ (\ref{eq3-is}), one finds $|T(\rho )|=\sum_kw_k%
\big|U_k\rho \,U_k^{*}\big|$ for $\rho \in {\mathcal{B}}_1({\mathcal{H}})$, and
so 
\begin{equation*}
\big\Vert{{\widehat{T}}(\rho )}\big\Vert_1\ =\ \text{tr}\bigl[\,|T(\rho )|\,%
\bigr]\ =\ \sum_kw_k\,\left\| \rho \right\| _1\ =\ \left\| \rho \right\| _1.
\end{equation*}
Conversely, assume $\widehat{T}$ is a trace preserving isometry on 
$\mathcal{B}_1({\mathcal{H}})$, and let $T$ denote its restriction to 
${\mathcal{B}}_1({\mathcal{H}})_s$. 
$T$ is clearly trace preserving and isometric and
therefore, by Lemma \ref{le2-stoch}, a stochastic isometry.%
\end{proof}

\begin{Prop}
\label{pro6-his}Let $T=\sum_kw_kT_k:V\rightarrow \tilde{V}$ be a stochastic
isometry, with $T_k$ pure stochastic isometries as in Eq.\ (\ref{eq3-is}).
Then 
\begin{equation*}
\tilde{T}(\tau )\ :=\ \sum_k\sqrt{w_k}\,T_k(\tau ),\ \ \tau \in 
{\mathcal{B}}_2({\mathcal{H}}),
\end{equation*}
defines an isometry $\tilde{T}$ of the Hilbert Schmidt class 
${\mathcal{B}}_2({\mathcal{H}})$.
\end{Prop}

\begin{proof}
For $\tau \in {\mathcal{B}}_2({\mathcal{H}})$, one easily verifies that 
\begin{equation*}
\tilde{T}(\tau )^{*}\tilde{T}(\tau )\ =\ \sum_kw_kT_k(\tau )^{*}T_k(\tau )\
\in {\mathcal{B}}_1({\mathcal{H}})
\end{equation*}
(where we used the orthogonality $T_k(\tau )^{*}T_\ell (\tau )=0$, $k\ne
\ell $), and so 
\begin{equation*}
\big\Vert{R(\tau )}\big\Vert_2^2\ =\ \sum_kw_k\text{tr}\bigl[U_k\tau
^{*}\tau \,U_k^{*}\bigr]\ =\ \big\Vert{\tau }\big\Vert_2^2.
\end{equation*}
Therefore, $\Vert {R(\tau )}\Vert _2=\Vert {\tau }\Vert _2$.%
\end{proof}

\begin{Rem} Let $T=\sum_kw_kT_k$ be a stochastic
isometry, with pure stochastic isometric components $T_k$, $\widehat{T}%
=\sum_kw_k{\widehat{T}}_k$ its extension to 
${\mathcal{B}}_1({\mathcal{H}})$. It
follows that $^{\mathrm{o}}T:=\sum_k{\widehat{T}}_k$ is an isometry on $%
{\mathcal{B}}({\mathcal{H}})$ and a Jordan $*$-isomorphism, and that its
restriction to the orthocomplemented lattice of projections of 
${\mathcal{H}}$, 
$\mathcal{P}({\mathcal{H}})$, is a lattice- and ortho-isomorphism from $%
\mathcal{P}({\mathcal{H}})$ onto the lattice 
\begin{equation*}
{\mathcal{P}}_T:=\left\{ \Pi \left( T(\rho )\right) \,:\,\rho \in 
{\mathcal{B}}_1({\mathcal{H}})_s\right\} .
\end{equation*}
The map $^{\mathrm{o}}T$ decomposes in a unique way into 
\begin{eqnarray*}
{^o}T{\ } &=&{\ {^o}}T{^{(L)}\,+\,{^o}}T{^{(A)},} \\
{^o}T{^{(L)}(a)\ } &=&{\ P^{(L)}\,{^o}}T{(a)\,P^{(L)},\ \ {^o}}T{^{(A)}(a)\
=\ P^{(A)}\,{^o}}T{(a)\,P^{(A)},\ \ a\in {\mathcal{B}}({\mathcal{H}})},
\end{eqnarray*}
and ${^oT}^{(L)}$ [${^oT}^{(A)}$] is a $*$-isomorphism [$*$%
-anti-isomorphism] of the $C^{*}$-algebra ${\mathcal{B}}({\mathcal{H}})$ onto
the subalgebras ${^oT}^{(L)}({\mathcal{B}}({\mathcal{H}}))$ [${^oT}^{(A)}(%
{\mathcal{B}}({\mathcal{H}}))$]. This is an illustration of the result of
Kadison \cite{Kad2} cited in Section 1.

Further we note that the projections $P^{(L)},P^{(A)}$ are in the centre of
the von Neumann algebra generated by $^oT\left( {\mathcal{B}}({\mathcal{H}}%
)\right) $. In physical terms, they induce a superselection rule on the
state space $T({{\mathcal{B}}}_1({{\mathcal{H}}})_s)$.
\end{Rem}

\begin{Rem}
Consider a completely positive $m$-mixing
stochastic isometry, $T=\frac 1m\sum_{k=1}^mT_k$, with $T_k(\rho )=U_k\rho
\,U_k^{*}$, all $U_k$ unitary (hence $m^{(L)}=m)$. The maps $U_k$ are not
uniquely determined by $T$, even apart from a phase factor, if $m\ge 2$.
This is apparent from the construction of the pure isometries $T_k$ in
Proposition \ref{pro5-mix4}, which was based on the choice of some nonzero
vector $\varphi ^0\in {\mathcal{H}}$ and an arbitrary \textsc{onb} $\varphi _k$
of $\Pi _{\varphi ^0}{\mathcal{H}}$. Accordingly, it can be shown that if $%
T(\rho )=\frac 1m\sum_kU_k\rho \,U_k^{*}=\frac 1m\sum_kV_k\rho \,V_k^{*}$,
where $V_k$ is another set of unitary maps, then 
\begin{equation*}
V_k\ =\ \sum_\ell \gamma _{k\ell }\,U_\ell ,
\end{equation*}
with $(\gamma _{k\ell })$ a unitary $m\times m$ matrix. It is easily
verified that this condition ensures that the $V_k$ define the same mixing
stochastic isometry as the $U_k$. The totality of all projections $%
Q_k=V_kV_k^{*}$ thus obtained from a given set $U_k$ using all unitary
matrices $(\gamma _{k\ell })$ commute with all elements of $T\left( \mathcal{%
B}_1({\mathcal{H}})\right) $ and ${^oT}\left( {\mathcal{B}}({\mathcal{H}})\right) $%
.
\end{Rem}

\begin{Rem}
The relation $\Pi _{\left. \varphi \right.
}\Pi _{\left. \psi \right. }\Pi _{\left. \varphi \right. }=\alpha \Pi
_{\left. \varphi \right. }$, valid for $m$-mixing stochastic isometries
shows that the set of projections $\Pi _{\left. \varphi \right. }$ bears
some fundamental similarities with the set of one dimensional projections.
In fact pairs of projections satisfying such a relation (with nonzero factor 
$\alpha $) are called \textsl{isoclinic}; their geometric significance has
been studied by von Neumann \cite{Mae}. Note that also the projections $%
P_k,Q_\ell $ discussed in Remark 6.2 are isoclinic.
\end{Rem}

A physically interesting feature of a stochastic isometry $T$ lies in the
fact that its inverse on the range $T(V)$ can be extended to a stochastic
map on $V$.

\begin{Thm}
\label{th6-inv}Let $T:V\to \tilde{V}$ be a stochastic isometry. The inverse $%
T^{-1}:T(V)\to V$ admits an extension to a stochastic map on $\tilde{V}$ as
follows. For $T$ expressed in the form of Eq. (\ref{eq3-is}), let $P_0$ be
the projection onto $\tilde{{\mathcal{H}}}_0$. The following defines a
stochastic map on $\tilde{V}$: 
\begin{equation*}
S(\rho )=\sum_{k=1}^NU_k^{*}\rho \,U_k\,+\,P_0\rho P_0,\ \ \rho \in \tilde{V}%
.
\end{equation*}
Then $T^{-1}=S|_{_{T(V)}}$.
\end{Thm}

\begin{proof}
The positivity of $S$ is obvious. Trace preservation follows from the fact
that $U_kU_k^{*}=P_k$ is the projection onto $\tilde{{\mathcal{H}}}_k$ and $%
\sum_{k=0}^NP_k=I_{\tilde{{\mathcal{H}}}}$.

Let $\rho =T(\sigma )$. Then, since $U_k^{*}U_\ell =\delta _{k\ell }I_{%
{\mathcal{H}}}$ and $P_0\,\rho \,P_0=0$, one has 
\begin{equation*}
S(\rho )=S\left( T(\sigma )\right) =\sum_{\ell =1}^N\sum_{k=1}^Nw_kU_\ell
^{*}U_k\,\sigma \,U_k^{*}U_\ell =\sum_{k=1}^Nw_k\sigma =\sigma =T^{-1}\left(
\rho \right) .
\end{equation*}
Hence $S|_{_{T(V)}}=T^{-1}$.%
\end{proof}

\noindent  This result suggests the possibility that reversible physical
state changes are not necessarily represented by surjective, and hence pure,
stochastic isometries, but that a state change effected by \textsl{any}
stochastic isometry is reversible: there exists a single dynamical map that
sends all final states back to the initial states. This interpretation is
further elaborated in \cite{BuQ}.

Pursuing further the dynamical interpretation of stochastic isometries, it
may be observed that such maps induce a reduction of the symmetries of the
physical system in question. In this sense it can be said that dynamical
maps represented as stochastic isometries describe the formation of \textsl{%
structure} (in a reversible way). According to a theorem due to Wigner, any
symmetry operation \cite{Bar}, defined as an angle preserving map of the set
of rays of ${\mathcal{H}}$ onto itself, is induced by a unitary or antiunitary
map according to Eq.\ (\ref{eq3-pure}). For simplicity, we consider a
completely positive stochastic isometry of the form (\ref{eq3-is}), with all 
$U_k$ unitary. Let $U$ be a unitary symmetry operation. Then we compute: 
\begin{equation*}
T\left( U\rho \,U^{*}\right) \ =\ \sum_{k=1}^Nw_k\left( U_kUU_k^{*}\right)
\,U_k\rho \,U_k^{*}\,\left( U_kU^{*}U_k^{*}\right) \ =\ {\tilde{U}}\,T(\rho
)\,{\tilde{U}}^{*},
\end{equation*}
where 
\begin{equation*}
{\tilde{U}}\ =\ \sum_kU_kU\,U_k^{*}\ =\ {^o}T(U)
\end{equation*}
is unitary. The map $U\longmapsto {\tilde{U}}$ is injective. But to every $U$
there do exist unitary maps $V\ne \tilde{U}$ such that 
\begin{equation}
{\tilde{U}}\,T(\rho )\,{\tilde{U}}^{*}=V\,T(\rho )\,V^{*}  \label{eq7-uni}
\end{equation}
for all $\rho \in {\mathcal{H}}_1({\mathcal{H}})$. In fact, define a unitary map 
$W=\sum_{k=0}^N\lambda _kP_k$, with $P_k=U_kU_k^{*}$, $\lambda _k\in \Bbb{C}$%
, $|\lambda _k|=1$, such that $W\ne I$. Let $V={\tilde{U}}W$. Then (\ref
{eq7-uni}) holds due to the fact that $W$ commutes with all $T(\rho )$. This
shows that on the state space $T\left( {\mathcal{H}}_1({\mathcal{H}})_s\right) $
not all symmetries of ${\mathcal{H}}$ can be distinguished.

There is another interpretation of a stochastic isometry $T$ and its
associated map $^{\mathrm{o}}T$ defined in Remark 6.1. These maps lead to a
physically equivalent description of all states and observables of the given
quantum system in the following sense: for all states $\rho $ and effects $a$%
, the corresponding states $T(\rho )$ and effects ${^o}T(a)$ give the same
probabilities: 
\begin{equation*}
\langle T{(\rho )}\,,\,{^o}T{(a)}\rangle \ =\ \langle {\rho }\,,\,{a}\rangle
.
\end{equation*}
This observation can be elaborated into a general theory of extensions of
the statistical description of a quantum physical system on the basis of
Theorem \ref{th3-is}\cite{Bus}.

\appendix 

\section{Algebraic Proof of Proposition \ref{pro3-mix}}

The decomposition of a mixing stochastic isometry into pure stochastic
isometries can be obtained by application of Kadison's work on isometries of
operator algebras \cite{Kad2}, thereby bypassing the explicit geometric
constructions of Section 5. Here we sketch the required steps.

Let $T:V\rightarrow \tilde{V}$ be an $m$-mixing stochastic isometry. One
first shows that the map ${^o}T:V\rightarrow \tilde{V}$, ${^o}T=m\,T$
extends to a linear, $*$-preserving, norm-bounded map ${^o}\widehat{T}{:}%
{\mathcal{B}}({\mathcal{H}})\rightarrow 
{\mathcal{B}}(\tilde{\mathcal{H}})$. This
map sends projections to projections and is $\sigma $-ortho-additive on the
set of projections.

Following an argument of Wright \cite{Wri}, ${^o}\widehat{T}$ is finally
shown to be a Jordan $*$-homomorpism.

The decomposition of ${^o}\widehat{T}$, and hence of $T$, is then obtained
by application of the arguments of Wright \cite{Wri}, which make use of
Kadison's theorem on isometries \cite{Kad2} in the form presented by Emch (%
\cite{Em}, Theorem 1, p.~ 153).

\end{document}